\documentclass[journal,twoside,web]{ieeecolor}
\usepackage{generic}
\usepackage{cite}
\usepackage{amsmath,amssymb,amsfonts}
\usepackage{graphicx}
\usepackage{hyperref}
\hypersetup{hidelinks=true}
\usepackage{textcomp}
\def\BibTeX{{\rm B\kern-.05em{\sc i\kern-.025em b}\kern-.08em
    T\kern-.1667em\lower.7ex\hbox{E}\kern-.125emX}}
\usepackage{hyperref}
\hypersetup{
    colorlinks=true,
    linkcolor=blue,
    filecolor=magenta,      
    urlcolor=cyan,
}

\usepackage{algorithm, algorithmicx}
\usepackage[noend]{algpseudocode}

\definecolor{lightgray}{rgb}{0.74, 0.74, 0.74} 

\DeclareMathOperator*{\argmin}{argmin}

\DeclareMathOperator*{\diag}{Diag}

\newcommand{\norm}[1]{\left\lVert#1\right\rVert}
\newcommand{\abs}[1]{\left|#1\right|}

\usepackage{stmaryrd}
\newcommand{\IversonBracket}[1]{\llbracket #1 \rrbracket}

\newcommand{\brk}[1]{\left(#1\right)}

\newcommand{\bsq}[1]{\left[#1\right]}
\newcommand{\obsq}[1]{\left]#1\right[}
\newcommand{\bcur}[1]{\left\{#1\right\}}

\newcommand{\ceil}[1]{\lceil#1\rceil}

\newcommand{\A}{\mathcal{A}}
\newcommand{\F}{\mathcal{F}}

\newcommand{\C}{\mathcal{C}}

\renewcommand{\O}{\mathcal{O}}
\renewcommand{\P}{\mathcal{P}}
\renewcommand{\S}{\mathcal{S}}
\renewcommand{\H}{\mathcal{H}}
\newcommand{\I}{\mathcal{I}}
\newcommand{\J}{\mathcal{J}}

\newcommand{\T}{\mathcal{T}}

\newcommand{\Q}{\mathcal{Q}}
\newcommand{\R}{\mathcal{R}}
\newcommand{\X}{\mathcal{X}}

\newcommand{\U}{\mathcal{U}}
\newcommand{\NN}{\mathbb{N}}
\newcommand{\ZZ}{\mathbb{Z}}
\newcommand{\RR}{\mathbb{R}}

\DeclareMathAlphabet{\mymathbb}{U}{BOONDOX-ds}{m}{n}
\newcommand{\one}{\mymathbb{1}}
\newcommand{\zero}{\mymathbb{0}}

\newtheorem{theorem}{Theorem}[section]
\newtheorem{corollary}[theorem]{Corollary}
\newtheorem{lemma}[theorem]{Lemma}

\newtheorem{definition}{Definition}
\newtheorem{problem}{Problem}
\newtheorem{assumption}{Assumption}

\newtheorem{proposition}{Proposition}[section]
\newtheorem{remark}{Remark}[section]

\newcommand{\lightcomment}[1]{\Comment{{\color{lightgray}\small {#1}}}}

\newcommand{\blue}[1]{\textcolor{blue}{#1}}
\renewcommand{\baselinestretch}{0.96}

\begin{document}



\title{\LARGE{Control Invariant Sets for Neural Network Dynamical Systems and Recursive Feasibility in Model Predictive Control}}



\author{Xiao Li$^{*}$, Tianhao Wei$^{\dagger}$, Changliu Liu$^{\dagger}$, \IEEEmembership{Senior Member, IEEE}, \\ Anouck Girard$^{*}$, \IEEEmembership{Senior Member, IEEE}, Ilya Kolmanovsky$^{*}$, \IEEEmembership{Fellow, IEEE}
\thanks{$^{*}$Xiao Li, Anouck Girard, and Ilya Kolmanovsky are with the Department of Aerospace Engineering, University of Michigan, Ann Arbor, MI 48109, USA. {\tt\small \{hsiaoli, anouck, ilya\}@umich.edu}}
\thanks{$^{\dagger}$Tianhao Wei, and Changliu Liu are with the Robotics Institute, Carnegie Mellon University, PA 15213, USA. {\tt\small \{twei2, cliu6\}@andrew.cmu.edu}}
}

\maketitle

\begin{abstract}
Neural networks are powerful tools for data-driven modeling of complex dynamical systems, enhancing predictive capability for control applications. However, their inherent nonlinearity and black-box nature challenge control designs that prioritize rigorous safety and recursive feasibility guarantees.
This paper presents algorithmic methods for synthesizing control invariant sets specifically tailored to neural network based dynamical models. These algorithms employ set recursion, ensuring termination after a finite number of iterations and generating subsets in which closed-loop dynamics are forward invariant, thus guaranteeing perpetual operational safety. 
Additionally, we propose model predictive control designs that integrate these control invariant sets into mixed-integer optimization, with guaranteed adherence to safety constraints and recursive feasibility at the computational level.
We also present a comprehensive theoretical analysis examining the properties and guarantees of the proposed methods.
Numerical simulations in an autonomous driving scenario demonstrate the methods' effectiveness in synthesizing control-invariant sets offline and implementing model predictive control online, ensuring safety and recursive feasibility.

\end{abstract}

\begin{IEEEkeywords}
Neural Networks, Invariant Sets, Model Predictive Control, Recursive Feasibility
\end{IEEEkeywords}

\section{Introduction}\label{sec:intro}
\IEEEPARstart{N}{eural} Networks (NN), from early developments in Multi-Layer Perceptrons (MLP) \cite{hornik1989multilayer} to recent advances in recurrent NNs \cite{mandic2001recurrent} and physics-informed NNs \cite{raissi2019physics}, have been widely adopted for modeling dynamical systems due to their ability to learn complex nonlinear behaviors directly from data, reducing the reliance on explicit analytical models. 
Theoretically, NNs serve as universal function approximators, allowing for accurate representation of intricate dynamical systems \cite{hornik1989multilayer, cybenko1989approximation}. In practice, the availability of well-established NN training algorithms \cite{goodfellow2016deep} and platforms \cite{pytorch} facilitates their use for data-driven identification \cite{chen1990non}. 
This versatility has enabled their deployment in diverse fields, including autonomous vehicles for human behavior prediction \cite{li2024decision}, finance for stock market prediction \cite{kimoto1990stock}, and biology for protein structure and interaction discovery \cite{baek2021accurate}. 

Building on the predictive power of NN-modeled Dynamical Systems (NNDS), their integration into Model Preditive Control (MPC) and decision-making has significantly enhanced the responsiveness \cite{zhang2023modeling} and adaptability \cite{li2024decision} of control strategies, leading to improved practical performance \cite{salzmann2023real}. 
Those algorithms utilize NNDS to forecast system evolutions while optimizing control actions to achieve task-related objectives and enforce engineering constraints \cite{mayne2000constrained}. Among these constraints, safety is often paramount, ensuring that system states remain within a predefined safe region and avoid potentially hazardous or undesirable conditions. However, the derivation of formal safety guarantees in NNDS-based control remains a critical challenge due to the inherent nonlinearity and lack of interpretability of NNs \cite{huang2020survey}.

Literature addressing safety in NNDS spans both open-loop and closed-loop considerations. The safety verification of open-loop NNDS primarily focuses on reachability analysis or output range estimation, which maps a prescribed NNDS input domain to an output range and examines whether it remains within predefined safe regions. Techniques such as interval bound propagation \cite{mirman2018differentiable}, mixed-integer programming \cite{dutta2018output} and nested optimization schemes \cite{ruan2018reachability} have been developed for this purpose. However, open-loop approaches inherently lack corrective or feedback control mechanisms to regulate the NNDS's behavior, making their safety guarantees conservative.

To address this limitation, reachability analysis has been extended to closed-loop systems, demonstrating its ability to verify safety under predefined feedback control policies. Studies in this domain analyze the evolution of linear \cite{hu2020reach} and nonlinear \cite{huang2019reachnn} dynamical systems governed by fixed control policies, which are often learned NN controllers. Popular methods include interval analysis \cite{jafarpour2024efficient}, Bernstein polynomials \cite{huang2019reachnn}, and Semi-Definite Programming \cite{hu2020reach}. A comprehensive tutorial on reachability analysis, along with its transcription into mathematical programs and numerical algorithms, is provided in \cite{liu2021algorithms}. 
However, in this safety verification process, the safety of the NNDS is tied to a fixed control policy within a given horizon, which may limit adaptability to dynamic environments where online optimization for control is required. In contrast, our approach establishes Control Invariant Sets (CISs) \cite{blanchini1999set} for NNDS, ensuring the safety of the closed-loop system without dependence on a specific control law. 

In particular, CIS synthesis \cite{blanchini1999set, bertsekas1972infinite, kerrigan2000invariant} involves designing a generalizable control strategy offline that ensures safety indefinitely. This makes CIS synthesis a more general and challenging problem, particularly for NNDS. Researchers have developed function-based methods that prescribe CIS as level sets of control barrier/Lyapunov functions \cite{ames2019control} and safety indices \cite{liu2014control}. Meanwhile, set-based methods have been explored for piecewise affine discrete-time systems \cite{rakovic2004computation}, alongside function-based approaches for NNDS \cite{mazouz2022safety}, where subsequent online optimization for control is not discussed. Since function-based methods might have limited ability to handle multi-horizon prediction while ensuring the recursive feasibility of online optimization for control, our method focuses on set-based developments tailored for NNDS.

Additionally, similar to function-based methods \cite{wei2022safe, wei2022persistently}, we integrate CIS into MPC design, facilitating online optimization for control while ensuring recursive feasibility in real-time operation. In contrast to \cite{wei2022safe}, which requires the formation selection of safety index functions, restricting the representational flexibility. It is also limited to a one-step prediction horizon for control due to computational constraints. Our method introduces an algorithmic framework for CIS synthesis.
This framework accommodates more flexible set representations for the CIS and incorporates feasibility considerations within broader MPC formulations. 

In particular, the proposed method synthesizes CIS for discrete-time NNDS, modeled by MLP with ReLU activation, with state and control constraints.
The proposed algorithm has several distinguished features that represent the main contribution of this article:

\begin{itemize}
\item The proposed method synthesizes CIS through a nested set recursion and reduces the set recursion into subproblems handled by interval reachability analysis and mixed-integer linear optimization, ensuring computational efficiency. The individual optimization problem scales linearly with the number of neurons in NNDS and the dimensionality of state and control variables.
\item Theoretical properties of the algorithmic framework guarantee that the set recursion terminates in a finite number of steps, making the CIS synthesis process computationally feasible. 
\item The algorithmic framework employs a state-space quantization approach, enabling flexible set representations as unions of hyperboxes, and minimizing tuning complexity by reducing parameter adjustments to a single variable governing quantization resolution. This parameter balances the resulting CIS size with the theoretical maximum number of iterations required for synthesis.
\item The CISs are incorporated into generalized MPC designs as mixed-integer linear constraints, ensuring the closed-loop NNDS stays within the CIS indefinitely. The MPC optimization remains recursively feasible and can be warm-started with feasible solutions from the offline CIS synthesis process, while its complexity scales linearly with the number of neurons in the NNDS, the prediction horizon, and the dimensions of state and control variables.
\item The theoretical results are validated through simulations in a vehicle lane-keeping application, demonstrating the ability to ensure recursive feasibility and maintain computational times suitable for real-time implementation.
\end{itemize}

This article is organized as follows: In Sec.~\ref{sec:problem}, we review the theoretical foundations of set invariance and introduce the NNDS, along with the considered CIS synthesis and MPC problems. Sec.~\ref{sec:methodSyn} presents the algorithmic framework and optimization problems formulated to synthesize CIS for NNDS. In Sec.~\ref{sec:methodSynThm}, we establish the theoretical properties of the CIS synthesis process. Sec.~\ref{sec:methodCtrl} integrates the synthesized CIS into the MPC framework for online control optimization with established theoretical properties. Sec.~\ref{sec:results} provides experimental validation through offline CIS synthesis and online MPC implementation in a simulated lane-keeping application. Finally, conclusions are given in Sec.~\ref{sec:conclusion}.
\section{Problem}\label{sec:problem} 
In this section, we first review set invariance in Sec.~\ref{problem:prelim}.  
Next, in Sec.~\ref{problem:nndm}, we introduce the closed-loop NNDS and outline the two primary directions for this work.  
Throughout the paper, lower-case letters (e.g., $a, b, x, y, u$) are used to represent vectors, upper-case letters (e.g., $W, I$) denote matrices, and scripted letters (e.g., $\I, \T, \R, \A, \X, \U$) are reserved for sets or set-valued mappings.

\subsection{Preliminaries}\label{problem:prelim}

We consider a dynamical system represented by a discrete-time model,
\begin{equation}\label{eq:dynamicalSys}
    x_{k+1} = f(x_k, u_k),
\end{equation}
where $x_k\in \X \subseteq \RR^{n_x}$ is a vector of state variables and $u_k\in \U \subseteq \RR^{n_u}$ is a vector of control inputs. 

\begin{definition}[Forward Invariant Set \cite{blanchini1999set}]\label{def:forwardInv}
For~an autonomous system $x_{k+1} = f(x_k)$, a non-empty set $\S \subseteq \X$ is a Forward Invariant Set (FIS) if and only if, for all $x_0 \in \S$, the system trajectory satisfies $x_k \in \S$ for all $k \in \NN$.
\end{definition}

\begin{definition}[Control Invariant Set \cite{blanchini1999set}]\label{def:ctrlInv}
For a non-autonomous system $x_{k+1} = f(x_k, u_k)$, a non-empty set $\C \subseteq \X$ is a CIS if and only if there exists an admissible feedback control law $u_k = \pi(x_k) \in \U$ such that $\C$ is a FIS for the closed-loop system $x_{k+1} = f(x_k, \pi(x_k))$.
\end{definition}

\begin{remark}
The FIS and CIS for their respective dynamical systems are not necessarily unique. The largest sets that contain all the invariant sets that are constraint admissible are referred to as the Maximal FIS \cite{gilbert1991linear} and the Maximal CIS \cite{blanchini1994ultimate}, respectively.
\end{remark}

\begin{definition}[One-Step Returnable Set \cite{blanchini1994ultimate}]\label{def:returnable_set}
For a dynamical system $x_{k+1} = f(x_k, u_k)$, the One-Step Returnable Set $\R(\I, \T)$ is a set of all the states in a non-empty set $\I \subseteq \X$, from which there exists an admissible control that drives the system to a non-empty target set $\T \subseteq \X$, i.e.,  
\end{definition}
\begin{equation}\label{eq:returnable_set}
    \R(\I,\T):= \left\{x_k\in\I: \exists u_k\in \U,~ s.t.\;x_{k+1}\in \T \right\}.
\end{equation}

\begin{definition}[One-Step Returnable Subset]\label{def:returnable_subset}
For a dynamical system $x_{k+1} = f(x_k, u_k)$, a set $\Tilde{\R}(\I, \T)$ is a One-Step Returnable Subset of a non-empty set $\I \subseteq \X$ to a non-empty target set $\T \subseteq \X$ if and only if, for all $x_k \in \Tilde{\R}(\I, \T) \subseteq \I$, there exists an admissible control $u_k \in \U$ such that $x_{k+1} \in \T$.
\end{definition}

\begin{definition}[$i$-Step Admissible Set 
\cite{dorea1999b}]\label{def:i_th_adm_set}
For a dynamical system $x_{k+1} = f(x_k, u_k)$ and $i \in \ZZ^{\geq 0}$, the $i$-Step Admissible Set $\A_i(\T)$ is a set of all initial states in the state space $\X$, starting from which there exists an admissible control sequence of length $i$ that keeps the system state trajectory inside the non-empty target set $\T \subseteq \X$ for $i$ steps, i.e.,  
\end{definition}
\begin{equation}\label{eq:i_th_adm_set}
    \A_i(\T) := 
    \left\{
    x_0\in\X:  
    \begin{array}{c}
         \exists u_k \in \U,~k = 0,\dots, i-1,
         \\
         s.t.~  x_k \in \T,~k = 1,\dots, i
    \end{array}
    \right\}.
\end{equation}

\begin{definition}[$i$-Step Admissible Subset]\label{def:i_th_adm_subset}
For a dynamical system $x_{k+1} = f(x_k, u_k)$ and $i \in \ZZ^{\geq 0}$, a set $\tilde{\A}_i(\T) \subseteq \T$ is an $i$-Step Admissible Subset of a non-empty set $\T \subseteq \X$ if and only if, for all $x_0 \in \tilde{\A}_i(\T)$, there exists an admissible control sequence of length $i$, i.e., $u_k \in \U$, $k = 0, 1, \dots, i-1$, such that the system trajectory satisfies $x_k \in \T$ for all $k \in \bcur{1, 2, \dots, i}$.
\end{definition}

\begin{remark}\label{rmk:i_th_adm_subset}
Any subset of the One-Step Returnable Set $\Tilde{\R} \subseteq \R(\I, \T)$ and the $i$-Step Admissible Set $\Tilde{\A}_i \subseteq \A_i(\T)$, including the empty set, is itself a valid One-Step Returnable Subset and an $i$-Step Admissible Subset, respectively. Additionally, a $j$-Step Admissible Subset is always an $i$-Step Admissible Subset for any $j \geq i$.
\end{remark}

\subsection{Closed-Loop Neural Network Dynamical System}\label{problem:nndm}
\begin{figure}[ht!]
\begin{center}
\includegraphics[width=0.8\linewidth]{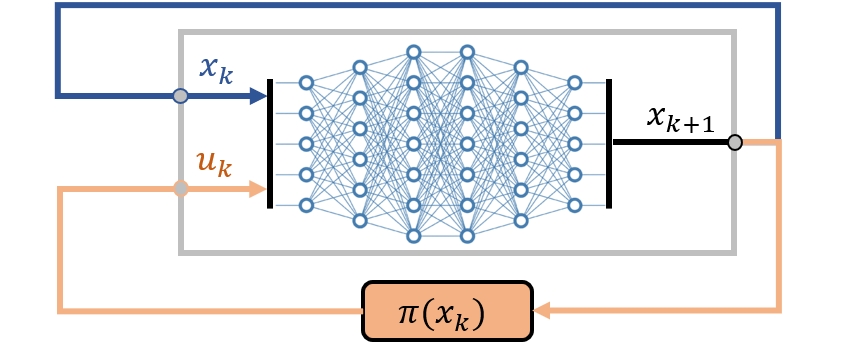}
\end{center}\vspace{-1.2em}
\caption{A schematic diagram of a closed-loop NNDS.}
\label{fig:nndm}
\end{figure}

As shown in Fig.~\ref{fig:nndm}, we consider a dynamical system represented by a pre-trained NNDS $f:\X\times\U\rightarrow\X$ that admits the following form of an MLP,
\begin{subequations}\label{eq:nn_structure}
\begin{equation}\label{eq:nn_structure:out}    
    x_{k+1} = W^{(\ell)}z^{(\ell-1)} +B^{(\ell)},
\end{equation}
\begin{equation}\label{eq:nn_structure:relu}        
    z^{(i)} = {\tt activate}^{(i)}(\hat{z}^{(i)}),
\end{equation}
\begin{equation}\label{eq:nn_structure:linear}     
    \hat{z}^{(i)} = W^{(i)}z^{(i-1)} +B^{(i)},
    ~
    i=1,\dots,\ell-1,
\end{equation}
\begin{equation}\label{eq:nn_structure:in}     
    z^{(0)} = [x_k^T, u_k^T]^T,
\end{equation}
\end{subequations}
where $\ell-1$ is the number of hidden layers, $n_i$, for $i=1,\dots,\ell-1$, represents the number of neurons in each hidden layer \eqref{eq:nn_structure:relu} and \eqref{eq:nn_structure:linear}, with $n_0 = n_x + n_u$ in the input layer \eqref{eq:nn_structure:in} and $n_\ell = n_x$ in the output layer \eqref{eq:nn_structure:out}, $W^{(i)}\in\mathbb{R}^{n_i\times n_{i-1}}$, $B^{(i)}\in\mathbb{R}^{n_i}$ and ${\tt activate}^{(i)}:\mathbb{R}^{n_i}\rightarrow\mathbb{R}^{n_i}$ are the weight matrix, the bias vector and an element-wise nonlinear activation function in the $i${th} layer, respectively. 
We assume that the nonlinear functions are ${\tt ReLU}$ activation functions
\begin{equation}\label{eq:relu}
    {\tt activate}^{(i)}(x) = \max\{0,x\},\;i=1,\cdots,\ell-1,
\end{equation}
that are commonly adopted in contemporary NN architectures and have demonstrated good empirical performance \cite{lecun2015deep}. This study addresses the following problems related to constrained control using NNDS. 

\begin{problem}[Control Invariant Set Synthesis]\label{problem:1}
Given the NNDS \eqref{eq:nn_structure}, synthesize offline a CIS $\C \subseteq \X_s$ such that the closed-loop system remains within $\C$ when starting from an arbitrary initial state $x_0 \in \C$, under some admissible feedback control law $u_k = \pi(x_k) \in \U$. Here, $\X_s \subseteq \X$ represents a safe set in the state space defined by a set of state constraints.
\end{problem}

\begin{problem}\label{problem:2}\textit{(Model Predictive Control and Recursive Feasibility)}
Consider the feedback control law as the solution to the following MPC Problem:
\begin{subequations}\label{eq:cocp} 
\begin{equation}\label{eq:cocp:obj} 
\pi(x_k) \in \argmin\limits_{\substack{u_{k+n},~x_{k+n+1}, \\ n = 0, 1, \dots, N-1}} K(x_{k+N}) + \sum_{n=0}^{N-1} l_{k+n}(x_{k+n}, u_{k+n}), 
\end{equation} 
\text{subject to:} \vspace{-1.5\baselineskip}
\begin{align} 
x_k = x_0,
\\
u_{k+n} \in \U,\label{eq:cocp:ctrl}
\\
x_{k+n+1} = f(x_{k+n}, u_{k+n}),\label{eq:cocp:dyna}
\\
x_{k+n+1} \in \X_s,~n = 0, 1, \dots, N-1   \label{eq:cocp:state}
\end{align} 
\end{subequations}
where $N$ is the prediction horizon length, $K$ is the terminal cost, $l_{k+n}$ is the running cost, and $x_0 \in \X_s$ is an initial state satisfying the constraints.
The MPC \eqref{eq:cocp} computes control laws $u_k = \pi(x_k)$ to regulate the state trajectory, as encoded by the equality constraints \eqref{eq:cocp:dyna}, while respecting control constraints \eqref{eq:cocp:ctrl} and state constraints \eqref{eq:cocp:state}. 
However, the MPC \eqref{eq:cocp} may not always be feasible and might lack an admissible control solution. To address this, additional measures are required to ensure the recursive feasibility of the MPC, i.e., if a feasible solution exists at a given time step $k$, it must also exist at every subsequent time step.
\end{problem}

\section{Control Invariant Set Synthesis}\label{sec:methodSyn}
In this section, we focus on addressing Problem~\ref{problem:1}. We first introduce assumptions and quantization of the state space in Sec.~\ref{method:syn:quant} as a preparation for numerical algorithms.  
Subsequently, we present numerical methods in Sec.~\ref{method:syn:algorithm} (also see Fig.~\ref{fig:algorithm_arch}) to synthesize a CIS for the closed-loop NNDS.
These approaches simplify Problem~\ref{problem:1} by reducing it to a more computationally tractable problem, which is then addressed using constrained optimization with integer variables and linear constraints, as detailed in Sec.~\ref{method:synthesis:propagation}. The theoretical properties of the proposed algorithms are summarized in subsequent~Sec.~\ref{sec:methodSynThm}. 

In the sequel, we use $a \leq b$ to denote the element-wise order between two vectors $a, b \in \RR^n$. Additionally, given a lower bound $a \in \mathbb{R}^n$, an upper bound $b \in \mathbb{R}^n$, and $a \leq b$, we use $[a, b]$ and $]a, b[$ to represent a closed and an open hyperbox in $\mathbb{R}^n$, respectively. In the following discussion, the term ``hyperbox" refers to closed ones unless otherwise~specified. 

\begin{figure*}[ht!]
    \centering
    \includegraphics[width=0.97\textwidth]{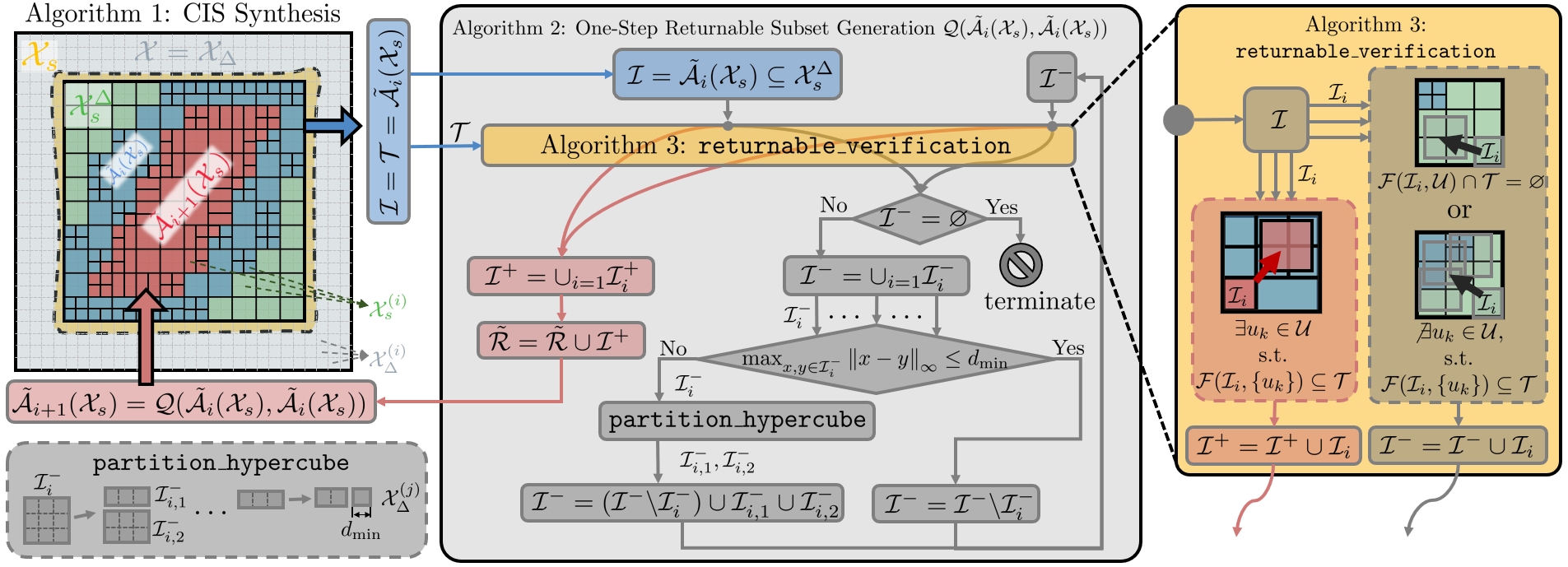}\vspace{-0.5em}
    \caption{A schematic diagram of the proposed numerical algorithms for synthesizing a CIS: Algorithm~\ref{al:synthesisInvSet} executes a set recursion that generates a sequence of $i$-Step Admissible Subsets and terminates when reaching a fixed point. The set recursion is advanced by Algorithm~\ref{al:synthesisUnderQ}, which computes a One-Step Returnable Subset $\Q(\tilde{\A}_i(\X_s), \tilde{\A}_i(\X_s))$ of the $i$-Step Admissible Subset $\tilde{\A}_i(\X_s)$. Algorithm~\ref{al:isReturnable} supports this process by verifying the one-step returnability of hyperboxes in the set $\tilde{\A}_i(\X_s)$ to itself.}
    \label{fig:algorithm_arch}
\end{figure*}

\subsection{Algorithmic Setup and State Space Quantization}\label{method:syn:quant}

We assume that the state space $\X$, safe set $\X_s\subseteq\X$, and admissible control set $\U$ satisfy Assumption~\ref{assume:x_u_box}. 

\begin{assumption}\label{assume:x_u_box}
We assume that the state space $\X$ and the set $\U$ are bounded hyperboxes that take the following form,
\begin{subequations}\label{eq:feasible_set}
\begin{equation}\label{eq:x_feasible_set}
    \X = [\underline{x}, \overline{x}]\in \H^{n_x},~\underline{x},\overline{x}\in\mathbb{R}^{n_x},~\underline{x}\leq\overline{x},
\end{equation}
\begin{equation}\label{eq:u_feasible_set}
    \U = [\underline{u}, \overline{u}]\in \H^{n_u},~\underline{u},\overline{u}\in\mathbb{R}^{n_u},~\underline{u}\leq\overline{u}.
\end{equation}
\end{subequations}
where $\H^{n}$ denotes the collection of hyperboxes in $\RR^n$. We assume the interior of the safe set, $\mathrm{int} \X_s$, is not empty.
\end{assumption}

We now define the following State Space Quantization (SSQ) process that discretizes the state space $\X$ into a union of hyperboxes, producing a collection of basis hyperboxes $\X_{\Delta}$ that support the quantization of the safe set $\X_s$ into many smaller safe hyperboxes $\X_s^{(i)}$, $i=1,\dots,n_s$.  
In particular, each safe hyperbox $\X_s^{(i)} \subseteq \X_{s}^{\Delta}$ can be represented as a union of basis hyperboxes.

\begin{definition}[\textbf{SSQ Process}]
    Given Assumption~\ref{assume:x_u_box}, $\X$ can be quantized into a finite union of disjoint \textbf{\textit{basis hyperboxes}} $\X_{\Delta}$, namely, 
    \begin{multline}\label{eq:x_feasible_set:quantize}
        \X= \cup_{i=1}^{n_\Delta} \X_{\Delta}^{(i)},
        ~
        \X_{\Delta}^{(i)} \in \X_{\Delta} \subseteq \H^{n_x},
        ~
        \X_{\Delta} :=
        \\
        \left\{
        \begin{array}{c}
            \bsq{a, b} = 
            \\
            \bigotimes_{j=1}^{n_x} [a_j, b_j]
        \end{array}
        :
        \begin{array}{c}
             a_j =\underline{x}_j + k d_{\min}, 
             ~
             b_j = 
             \\
             \min\{\underline{x}_j + (k+1) d_{\min}, \overline{x}_j\}, 
             \\
             k \in\{ 0,1,\dots, \ceil{\frac{(\overline{x}_j-\underline{x}_j)}{d_{\min}}} \}
        \end{array}
        \right\}
    \end{multline}
    where $n_{\Delta} = \prod_{j=1}^{n_x} \ceil{\frac{(\overline{x}_j-\underline{x}_j)}{d_{\min}}}$ is the number of basis hyperboxes in $\X_{\Delta}$, $\ceil{x}$ is the smallest integer greater than or equal to $x$, $d_{\min}$ is the SSQ resolution denoting the side length of the basis hyperboxes, $\bigotimes$ denotes the Cartesian product, and the subscript $j$ represents the $j$th element in the respective~vectors.   
\end{definition}
    
Afterward, our discussion focuses on sets that belong to the power set of $\X_{\Delta}$ under union, denoted as $\P_{\cup}(\X_{\Delta})$, which contains all possible unions of basis hyperboxes in $\X_{\Delta}$, i.e.,
\begin{equation}\label{eq:ssq:power_union}
    \P_{\cup}(\X_{\Delta}) := 
    \{ 
    \S
    :
    \exists \J \subseteq \bsq{n_\Delta},
    ~
    \S = \cup_{j\in\J} \X_{\Delta}^{(j)}
    \}
\end{equation}
where $\bsq{n_{\Delta}} := \{1,2,\dots,n_\Delta\}$ represents a set of indices smaller than or equal to $n_{\Delta}$.    
In addition, given an non-empty $\mathrm{int} \X_s$, the safe set $\X_s$ can be approximated from within, i.e., $\X_{s}^{\Delta}\subseteq \X_s$, by a finite union of safe hyperboxes $\X_{s}^{\Delta}$ (see the left side of Fig.~\ref{fig:algorithm_arch}) that obeys the following condition,
\begin{equation}\label{eq:x_safe:quantize}
    \begin{array}{c}            
        \X_{s}^{\Delta} := \bigcup_{i=1}^{n_s} \X_s^{(i)},
        \text{ s.t. }
        \X_s^{(i)}\in\H^{n_x}, 
        \X_s^{(i)} \subseteq \X_s,
        \\
        \X_s^{(i)} \in \P_{\cup}(\X_{\Delta}),
        ~
        \forall i =1,\dots,n_s,
    \end{array}
\end{equation}
where $n_s$ is the number of safe hyperboxes in $\X_{s}^{\Delta}$. We also note that $\X_{s}^{\Delta}\in \P_{\cup}(\X_{\Delta})$.  

\begin{remark}   
    We note that the inner approximation of the safe set $\X_s$, i.e., $\X_s^\Delta=\bigcup_{i=1}^{n_s} \X_s^{(i)}\subseteq\X_s$, is achieved using a coarser quantization, as described in \eqref{eq:x_safe:quantize}, compared to the SSQ process in \eqref{eq:x_feasible_set:quantize}. Namely, each safe hyperbox $\X_s^{(i)}\in \P_{\cup}(\X_{\Delta})$ can be larger than the basis hyperboxes (see Fig.~\ref{fig:algorithm_arch}).  
    This approach is motivated by memory considerations for $\X_s^{\Delta}$. In the case of a finer SSQ \eqref{eq:x_feasible_set:quantize} with a smaller $d_{\min}$, directly approximating $\X_s$ using smaller basis hyperboxes, i.e., making each safe hyperbox a basis hyperbox $\X_s^{(i)} \in \X_{\Delta}$ instead, can be memory-intensive.
    The coarser quantization \eqref{eq:x_safe:quantize} reduces memory demands while still allowing flexibility to increase the resolution of individual safe hyperbox $\X_s^{(i)}$ by further dividing it into smaller basis hyperboxes when necessary.
\end{remark}

In addition, due to the basis hyperboxes defined in \eqref{eq:x_feasible_set:quantize} being disjoint, we note that the set $\P_{\cup}(\X_{\Delta})$ forms an algebra of sets $(\X_{\Delta}, \P_{\cup}(\X_{\Delta}))$ that possesses the following properties.

\begin{proposition}\label{prop:ssq:unique}
    Consider the SSQ process, for all $\S\in \P_{\cup}(\X_{\Delta})$, there exists a unique set of indices $\J=\Delta(\S) \subseteq [n_\Delta]$ such that $\S = \cup_{j\in\J} \X_{\Delta}^{(j)}$, namely,
    \begin{equation*}
        \forall \S \in \P_{\cup}(\X_{\Delta}),
        ~
        \exists! \J = \Delta(\S) \subseteq \bsq{n_\Delta},
        ~
        \text{s.t. } \S = \cup_{j\in\J} \X_{\Delta}^{(j)},
    \end{equation*}
    where $\Delta(\cdot): \P_{\cup}(\X_{\Delta}) \to \NN$ extracts the indices of the basis hyperboxes that construct the set $\S$.
\end{proposition}

\begin{proof}
    The proof follows from the fact that the basis hyperboxes $\X_\Delta$ are disjoint.
\end{proof}

\begin{theorem}\label{thm:set_algebra}
    The power set of $\X_{\Delta}$ under union, $\P_{\cup}(\X_{\Delta})$, as defined in \eqref{eq:ssq:power_union}, forms \textit{\textbf{an algebra of sets}} $(\X_{\Delta}, \P_{\cup}(\X_{\Delta}))$, namely, it contains the empty set, is closed under complement in $\X_{\Delta}$, and is closed under binary unions.
\end{theorem}

\begin{proof}
    According to definition \eqref{eq:ssq:power_union}, it is evident that $\varnothing\in \P_{\cup}(\X_\Delta)$ since $\varnothing =\cup_{j\in\Delta (\varnothing)} \X_{\Delta}^{(j)}$ and $\Delta (\varnothing) = \varnothing \subseteq \bsq{n_\Delta}$. Meanwhile, for all sets $\S\in \P_{\cup}(\X_{\Delta})$, we have $\S=\cup_{j\in\Delta(\S)} \X_{\Delta}^{(j)}$. Then, due to the basis hyperboxes defined in \eqref{eq:x_feasible_set:quantize} being disjoint, the complement $\S^c = \X_{\Delta}\backslash \S \in \P_{\cup}(\X_{\Delta})$ since $\S^c= \cup_{j\in\Delta(\S^c)} \X_{\Delta}^{(j)}$ with an index set $\Delta(\S^c)= \bsq{n_\Delta}\backslash \Delta(\S) $, where $\backslash$ denotes the set difference. Lastly, if $\S_1=\cup_{j\in\Delta(\S_1)} \X_{\Delta}^{(j)}$, $\S_2=\cup_{j\in\Delta(\S_2)} \X_{\Delta}^{(j)}$ and $\Delta(\S_1),\Delta(\S_2)\subseteq \bsq{n_\Delta}$, then  $\S_1 \cup \S_2 \in \P_{\cup}(\X_{\Delta})$ since  $\S_1 \cup \S_2=\cup_{j\in\Delta(\S_1)\cup \Delta(\S_2)} \X_{\Delta}^{(j)}$.
\end{proof}

Then, the Corollary below follows from De Morgan's laws.

\begin{corollary}\label{coro:set_algebra}
    The algebra of sets $(\X_{\Delta}, \P_{\cup}(\X_{\Delta}))$ is closed under finite intersections, and finite unions.
\end{corollary}

\begin{remark}
    We also note that a more general basis hyperboxes $\X_{\Delta}$ can be constructed according to, 
    \begin{multline*}
        \X = \cup_{i=1}^{n_{\Delta}} \X_\Delta^{(i)},
        ~
        \X_\Delta^{(i)} \in \X_{\Delta}, 
        \text{ s.t.}
        \\
        \X_\Delta^{(i)}\cap\X_{\Delta}^{(j)} = \varnothing, 
        \forall i\neq j,~ i,j\in [n_\Delta]
        \\
        \text{and }
        \X_\Delta^{(i)} \in \H^{n_x},
        ~
        \forall i = 1,2,\dots,n_{\Delta}.
    \end{multline*}
    Here, we adopt an evenly spaced mesh grid quantization for a simpler code implementation.
\end{remark}

\subsection{Problem Reduction and Numerical Algorithms}\label{method:syn:algorithm}

Given an NNDS $f$ in \eqref{eq:nn_structure} and the SSQ process, we develop a computational procedure illustrated in Fig.~\ref{fig:algorithm_arch} and detailed in Algorithms~\ref{al:synthesisInvSet},~\ref{al:synthesisUnderQ}, and~\ref{al:isReturnable}. This method synthesizes a CIS $\C$ within the union of safe hyperboxes $\X_s^{\Delta}\subseteq \X_s$.

\begin{algorithm}[htp]
\caption{CIS Synthesis}
\label{al:synthesisInvSet}
\begin{algorithmic}[1]
\Require $\X=\X_\Delta$, $\U$, $\X_s^{\Delta}\subseteq\X_s$, $f(\cdot,\cdot)$
\lightcomment{the inputs are the state space $\X=\X_\Delta$ and control admissible set $\U$ according to Assumption~\ref{assume:x_u_box}; the quantized safe set $\X_s^{\Delta}$ according to the SSQ process; and the NNDS $f$}
\State $i = 0$, $\tilde{\A}_0(\X_s)=\X_s^{\Delta}$
\lightcomment{initialize a $0$-Step Admissible Subset using $\X_s^{\Delta}\subseteq\X_s$, which is trivially safe without the application of control}
\State Initialize $\tilde{\A}_1(\X_s)=\varnothing$
\lightcomment{initialize a $1$-Step Admissible Subset with an empty set according to Remark~\ref{rmk:i_th_adm_subset}}
\While{$\tilde{\A}_{i+1}(\X_s)\neq\tilde{\A}_{i}(\X_s)$}
\lightcomment{terminate the iterative algorithm once the $i$-Step Admissible Subset converges}
\State $\tilde{\A}_{i+1}(\X_s)=\Q\left(\tilde{\A}_i(\X_s), \tilde{\A}_i(\X_s) \right)$
\lightcomment{derive an $(i+1)$-Step Admissible Subset using a One-Step Returnable Subset of the $i$-Step Admissible Subset according to Algorithm~\ref{al:synthesisUnderQ}}
\State $i = i + 1$
\lightcomment{increase iteration index number by one}
\EndWhile
\\
\Return{$\C=\tilde{\A}_{i}(\X_s)$}
\lightcomment{ return a CIS}
\end{algorithmic}
\end{algorithm}

The key idea of Algorithm~\ref{al:synthesisInvSet} lies in the set recursion defined in line 4, which follows a structure similar to those adopted in \cite{bertsekas1972infinite, kerrigan2000invariant}. Here, the set-valued function $\Q(\I, \T)$ determines a One-Step Returnable Subset from the initial set $\I$ to the target set $\T$, ensuring that $\Q(\I, \T) \subseteq \R(\I, \T) \subseteq \I$. This implies $\tilde{\A}_{i+1}(\X_s) = \Q(\tilde{\A}_i(\X_s), \tilde{\A}_i(\X_s)) \subseteq \tilde{\A}_i(\X_s)$. 
Then, initialized with $\tilde{\A}_0(\X_s) = \X_s^{\Delta} \subseteq \X_s$, the set recursion generates a nested sequence of $i$-Step Admissible Subsets within the safe set $\X_s$, i.e., $\X_s \supseteq \X_s^\Delta = \tilde{\A}_0(\X_s) \supseteq \tilde{\A}_1(\X_s) \supseteq \tilde{\A}_2(\X_s) \supseteq \cdots$.
Intuitively, for the NNDS, initial states within $\Q(\tilde{\A}_i(\X_s), \tilde{\A}_i(\X_s))$ can be regulated to $\tilde{\A}_i(\X_s)\subseteq \X_s$ in the next step. From $\tilde{\A}_i(\X_s)$ which is an $i$-Step Admissible Subset, i.e., $\tilde{\A}_i(\X_s)\subseteq \A_i(\X_s)$, the system can remain within the safe set $\X_s$ for at least $i$ subsequent steps. This property ensures that $\tilde{\A}_{i+1}(\X_s) = \Q(\tilde{\A}_i(\X_s), \tilde{\A}_i(\X_s))$ is an $(i+1)$-Step Admissible Subset of the safe set $\X_s$, i.e., $\tilde{\A}_{i+1}(\X_s)\subseteq \A_{i+1}(\X_s)$.
Upon termination, the $i${th} and $(i+1)$-Step Admissible Subsets become equal, i.e., $\C = \tilde{\A}_i(\X_s) = \tilde{\A}_{i+1}(\X_s)$. By Definition~\ref{def:returnable_subset} and the relationship $\C = \Q(\C, \C)$, it follows that for all $x_k \in \C$, there exists an admissible control $u_k \in \U$ such that $x_{k+1} = f(x_k, u_k) \in \C$.  
This ensures that the closed-loop NNDS can remain within the set $\C = \tilde{\A}_i(\X_s) \subseteq \X_s$ indefinitely, rendering $\C$ a CIS.

Meanwhile, the set-valued function $\Q(\I, \T)$ is computed following Algorithm~\ref{al:synthesisUnderQ}. The primary challenge in computing this One-Step Returnable Subset lies in verifying the existence of an admissible control for each state vector $x_k \in \I$, which typically consists of infinitely many states, such that it can return to the target set $\T$ in one step.  
To address this, we note that the function $\Q$, defined according to Algorithm~\ref{al:synthesisUnderQ}, is only called by Algorithm~\ref{al:synthesisInvSet} with inputs $\I=\T = \tilde{\A}_i(\X_s)$, where $\tilde{\A}_i(\X_s)\in \P_{\cup}(\X_\Delta)$ is a union of hyperboxes (as described in later Remark~\ref{rmk:livesin_power_union}).  
Instead of verifying each state individually, we verify each hyperbox $\I_i$ within the quantized $i$-Step Admissible Subset $\tilde{\A}_i(\X_s)$, i.e., $\I_i \subseteq \tilde{\A}_i(\X_s) \in \P_{\cup}(\X_\Delta)$. The goal is to determine whether there exists a common admissible control $u_k \in \U$ such that, for any initial state $x_k \in \I_i$, the system can return to the target set $\T = \tilde{\A}_i(\X_s)$ in the next step using this single $u_k$. This ensures that $\I^+_i \subseteq \R(\I , \T)$ for all $\I^+_i \subseteq \I^+\subseteq \I$ in Algorithm~\ref{al:synthesisUnderQ}.

\begin{algorithm}[htp]
\caption{One-Step Returnable Subset Generation $\Q$}
\label{al:synthesisUnderQ}
\begin{algorithmic}[1]
\Require $\I=\tilde{\A}_i(\X_s)$, $\T=\tilde{\A}_i(\X_s)$, $d_{\min}$
\lightcomment{the inputs are an $i$-Step Admissible Subset $\tilde{\A}_i(\X_s)\subseteq\X_s$ as a union of hyperboxes, and parameter $d_{\min}$ defined in the SSQ process}
\State $\tilde{\R} = \varnothing$
\lightcomment{initialize Returnable Subset using empty set}
\State $\I^+$, $\I^- = {\tt returnable\_verification} (\I , \T)$
\lightcomment{check each hyperbox $\I_i\subseteq\I$ using Algorithm~\ref{al:isReturnable}, and verify if $\exists u_k \in \U$ such that $\I_i$ can be made a One-Step Returnable Subset of $\I$ to the target set $\T$, i.e., $\I_i\subseteq\R(\I ,\T)\subseteq\R(\X_s ,\T)$. If this holds true, store $\I_i$ in the hyperbox union $\I^+ = \cup_{i} \I_i^+$; otherwise, store it in $\I^- = \cup_{i}\I_i^-$, s.t. $\I^+ \cup \I^- = \I$ }
\State $\tilde{\R} =\tilde{\R}\bigcup\I^+$ 
\lightcomment{update Returnable Subset}
\While{$\I^-\neq \varnothing$}
\lightcomment{terminate iteration if the union $\I^-$ contains no hyperboxes}
\For{$\I^-_i\subseteq\I^-$}
\lightcomment{check each hyperbox in union}
\If{$\max_{x,y\in\I^-_i}{\norm{x-y}_{\infty}} \leq d_{\min}$} 
\lightcomment{check if a hyperbox has been reduced to one of the basis hyperboxes, i.e., $\exists j\in[n_\Delta]$ s.t. $\I^-_i=\X_{\Delta}^{(j)}$, by verifying if its largest side length among all dimensions is smaller than $d_{\min}$}
\State $\I^-=\I^-\backslash\I^-_i$ 
\lightcomment{remove the hyperbox $\I^-_i$ from the union $\I$ if it has been reduced to one of the basis hyperboxes}
\Else
\State $\I^-_{i,1}$, $\I^-_{i,2} = \tt{partition\_hyperbox}(\I^-_i)$ 
\lightcomment{partition hyperbox $\I^-_i$ into two smaller hyperboxes $\I^-_{i,1}, \I^-_{i,2}\in \P_{\cup}(\X_{\Delta})$ such that $\I^-_i = \I^-_{i,1} \cup \I^-_{i,2}$ and $\I^-_{i,1} \cap \I^-_{i,2}=\varnothing$}
\State $\I^-=(\I^-\backslash \I^-_i)\cup\I^-_{i,1}\cup\I^-_{i,2}$
\lightcomment{remove $\I^-_i$ from the union and add $\I^-_{i,1}$ and $\I^-_{i,2}$ to the union}
\EndIf
\EndFor
\State $\I^+$, $\I^- = {\tt returnable\_verification}(\I^-, \T)$
\lightcomment{repeat the verification procedure described in line 2 with the union of hyperboxes $\I^-$. Similarly, update the unions $\I^+$ and $\I^-$ accordingly}
\State $\tilde{\R} =\tilde{\R}\bigcup\I^+$
\lightcomment{update Returnable Subset}
\EndWhile
\\
\Return{$\Q(\I, \T) = \tilde{\R}$}
\lightcomment{$\tilde{\R}\subseteq  \R(\I,\T)\subseteq  \R(\X_s,\T)$}
\end{algorithmic}
\end{algorithm}

Additionally, finding a common admissible control for all states within a hyperbox can be challenging. To address this, each hyperbox $\I_i^- \subseteq \I^-\subseteq \I=\tilde{\A}_i(\X_s)$ that fails the one-step returnability verification in the previous iteration is partitioned, as described by the function ${\tt partition\_hyperbox}$ in line 9 of Algorithm~\ref{al:synthesisUnderQ}. This process reduces the size of the hyperbox (also see Fig.~\ref{fig:algorithm_arch}) and generates a finer quantized $i$-Step Admissible Subset $\tilde{\A}_i(\X_s)$ with smaller hyperboxes $\I^-_{i,1},\I^-_{i,2}$ in $\P_{\cup}(\X_{\Delta})$.  
The user-defined variable $d_{\min}$ controls the resolution of the SSQ process that supports this refinement process. The algorithm iteratively refines the quantization and re-verifies from lines 4 to 12 in Algorithm~\ref{al:synthesisUnderQ}. It terminates when each hyperbox $\I_i^- \subseteq \I^-$ matches one of the basis hyperboxes $\X_{\Delta}^{(j)} \subseteq \X_{\Delta}$ with a SSQ resolution $d_{\min}$ and still fails the verification, namely, $\forall \I_i^- \subseteq \I^-, \exists j\in[n_\Delta] \text{ s.t. } \I^-_i=\X_{\Delta}^{(j)}$.
Then, Algorithm~\ref{al:synthesisUnderQ} produces a One-Step Returnable Subset as the union of successfully verified hyperboxes $\I^+_i \subseteq \I^+$.

\begin{algorithm}[htp]
\caption{${\tt returnable\_verification}$}
\label{al:isReturnable}
\begin{algorithmic}[1]
\Require $\I=\cup_{i}~\I_i$, $\T=\tilde{\A}_i(\X_s)$
\lightcomment{the inputs are a union of hyperboxes $\I$ to be verified and a target set $\T$.}
\State $\I^+=\varnothing$, $\I^-=\varnothing$
\lightcomment{initialize the unions of hyperboxes: the verified One-Step Returnable Subset ($\I^+$) and the non-determined union ($\I^-$)}
\For{$\I_i\subseteq \I$}
\lightcomment{examine each hyperboxes}
\If{\blue{$\bar{\F}(\I_i, \U)\cap\T=\varnothing$}}
\lightcomment{given $\F(\I_i, \U)\subseteq \bar{\F}(\I_i, \U)$, an empty intersection between set $\bar{\F}(\I_i, \U)$ and the target set $\T$ implies that $\F(\I_i, \U)\cap\T=\varnothing$ and there is no $x_k \in \I_i$ and $u_k \in \U$ such that $f(x_k, u_k) \in \T$.}
\State $\I^-=\I^-\cup\I_i$
\lightcomment{$\I_i \cap \R(\I,\T) = \varnothing$}
\Else
\If{\blue{$\exists u_k\in\U$ such that $\F(\I_i, \{u_k\})\subseteq \T$}}
\lightcomment{$\I_i$ is a One-Step Returnable Subset by the Definition~\ref{def:returnable_subset}. To this point, Problem~\ref{problem:1} is reduced to Problem~\ref{problem:1:reduced}.}
\State $\I^+=\I^+\cup\I_i$
\lightcomment{$\I_i \subseteq \R(\I,\T)$}
\Else
\State $\I^-=\I^-\cup\I_i$ 
\lightcomment{undetermined. It is possible that $\I_i \cap \R(\I, \T) \neq \varnothing$, or $\I_i \cap \R(\I, \T) = \varnothing$}
\EndIf
\EndIf
\EndFor
\\\Return{$\I^+$, $\I^-$}
\end{algorithmic}
\end{algorithm}

Afterward, the method used to verify the one-step returnability of the union of hyperboxes, i.e., whether a hyperbox $\I_i \subseteq \R(\I, \T)$ where $\I_i\subseteq\I$, is summarized in Algorithm~\ref{al:isReturnable}. 
Here, we utilize two set-valued mappings, $\F: \X \times \U \rightarrow \X$ and $\bar{\F}: \X \times \U \rightarrow \X$, in lines 3 and 6 of Algorithm~\ref{al:isReturnable}, respectively.
The mapping $\bar{\F}$ performs the Reachability Analysis such that the following condition holds,
\begin{equation}
    \F(\X_k, \U_k)\subseteq \bar{\F}(\X_k, \U_k),
    ~
    \forall \X_k\subseteq \X, 
    ~
    \forall \U_k\subseteq \U,
\end{equation}
where $\F: \X \times \U \rightarrow \X$ derives the exact one-step reachable set of the NNDS, defined as:  
\begin{equation}
\F(\X_k, \U_k):=
\left\{
x_{k+1} \in \X : 
\begin{array}{r}
x_{k+1} = f(x_k, u_k), 
\\
\forall x_k \in \X_k, \forall u_k \in \U_k
\end{array}
\right\}.
\end{equation} 
They are defined for different use cases in Algorithm~\ref{al:isReturnable}:

\begin{itemize}
    \item First of all, to reduce the computational burden, we filter out hyperboxes $\I_i$ in line 3 for which the system, with initial states $x_k \in \I_i$, cannot reach the target set $\T$ in the next step. This is realized using a set-valued mapping $\bar{\F}$.  
    In particular, as detailed in line 3 of Algorithm~\ref{al:isReturnable}, $\bar{\F}$ computes a superset $\bar{\F}(\I_i, \U)$ of the actual reachable set $\F(\I_i, \U)$ such that $\F(\I_i, \U)\subseteq \bar{\F}(\I_i, \U)$. This procedure helps to eliminate hyperboxes that are guaranteed not to reach the target set $\T$ in one step. 
    \item In contrast, in line 6, the algorithm searches for the existence of an admissible control $u_k \in \U$ such that $\F(\I_i, \{u_k\}) \subseteq \T$, where $\{u_k\}$ denotes a singleton set containing only one element $u_k$. Here, the set-valued function $\F$ propagates the sets $\I_i$ and $\{u_k\}$ through the dynamics $f$, where the control variable $u_k$ is not fixed and its existence is unknown. This challenge falls beyond the scope of traditional Reachability Analysis and its off-the-shelf tools.
\end{itemize}

Eventually, by introducing the algorithms in this section, Problem~\ref{problem:1} is reduced to more tractable problems, summarized in Problem~\ref{problem:1:reach} and Problem~\ref{problem:1:reduced}, which correspond to the two bullet points mentioned above. These problems are addressed in the subsequent Sec.~\ref{method:synthesis:propagation}.

\begin{problem}[Reachability Analysis]\label{problem:1:reach}  
Given a hyperbox $\I_i = [\underline{x}_k, \overline{x}_k]$ and $\U = [\underline{u}, \overline{u}]$ assumed in Assumption~\ref{assume:x_u_box}, compute a superset $\bar{\F}(\I_i, \U)$ of the one-step reachable set $\F(\I_i, \U)$ such that $\F(\I_i, \U)\subseteq \bar{\F}(\I_i, \U)$.
\end{problem}  

\begin{problem}[Decision-Variable-Dependent Set Propagation]\label{problem:1:reduced}  
Given a hyperbox $\I_i = [\underline{x}_k, \overline{x}_k]\subseteq \X$ and a target set $\T = \tilde{\A}_i(\X_s)$, which is a finite union of closed hyperboxes derived from Algorithm~\ref{al:synthesisUnderQ}, verify whether the hyperbox $\I_i$ is one-step returnable to the target set $\T$ with a common admissible control $u_k \in \U$, i.e., verify if there exists $u_k \in \U$ such that $\F(\I_i, \{u_k\}) \subseteq \T$.
\end{problem}  

\begin{remark}\label{rmk:livesin_power_union}
    Based on the algorithmic construction, Assumption~\ref{assume:x_u_box}, the SSQ process, Theorem~\ref{thm:set_algebra}, and Corollary~\ref{coro:set_algebra}, the outputs and inputs in Algorithms~\ref{al:synthesisInvSet}, \ref{al:synthesisUnderQ}, and \ref{al:isReturnable} are restricted to being either hyperboxes, unions of hyperboxes, or empty sets.  
    In particular, the algorithmic process is initialized with $\tilde{\A}_0(\X_s) = \X_s^{\Delta} \in \P_{\cup}(\X_\Delta)$ in Algorithm~\ref{al:synthesisInvSet}. The built-in function $\Q$, utilizing Algorithm~\ref{al:synthesisUnderQ}, operates on unions of hyperboxes, specifically $\I^+$, $\I^-$, and $\tilde{\R}$, all of which belong to the set $\P_{\cup}(\X_\Delta)$.  
    Subsequently, the resulting sequence of $i$th Admissible Subsets satisfies $\tilde{\A}_i(\X_s) \in \P_{\cup}(\X_\Delta)$ for all $i \in \mathbb{Z}$.
\end{remark}

\begin{remark}
The mapping $\bar{\F}(\I_i, \U)$ is used because the exact computation of $\F(\I_i, \U)$ can be difficult. Instead, using interval arithmetic, Bernstein polynomials \cite{huang2019reachnn}, or Semi-Definite Programming \cite{hu2020reach}, a superset of $\F(\I_i, \U)$, denoted as $\bar{\F}(\I_i, \U)$, can be derived.
\end{remark}


\subsection{Sub-problems: Reachability Analysis and Decision-Variable-Dependent Set Propagation}\label{method:synthesis:propagation}

In this section, we address aforementioned Problem~\ref{problem:1:reach} and \ref{problem:1:reduced}. Firstly, given hyperboxes $\I_i = [\underline{x}_k, \overline{x}_k]$ and $\U = [\underline{u}, \overline{u}]$ in Problem~\ref{problem:1:reach}, we compute the set $\bar{\F}(\I_i, \U)$ as a hyperbox with mapping $\bar{\F}$ that admits the following form,
\begin{subequations}\label{eq:reachability}
\begin{equation}\label{eq:reachability:Fbar}
\begin{array}{c}
    \bar{\F}(\X_k, \U_k) := [\hat{a}^{(\ell)}, \hat{b}^{(\ell)}],
    \\
    \forall 
    \X_k:=[x_l, x_u], 
    ~
    \forall
    \U_k:=[u_l, u_u],
\end{array}
\end{equation}
\begin{equation}\label{eq:reachability:interval_arithmetic}
    \bcur{[\hat{a}^{(i)}, \hat{b}^{(i)}]}_{i=1}^{\ell}
    = {\tt reachable\_hyperboxes} (\X_k, \U_k)
\end{equation}
\end{subequations}
where function ${\tt reachable\_hyperboxes}$ derives a collection of hyperboxes using the following arithmetic operations,
\begin{subequations}\label{eq:interval_arithmetic}
\begin{equation*}
\begin{array}{c}
    {\tt reachable\_hyperboxes} (\X_k, \U_k) := 
    \\
    \overbrace{\hspace{0.8\linewidth}} \hspace{0.07\linewidth}
\end{array}
\end{equation*}
\vspace{-1em}
\begin{equation}\label{eq:interval_arithmetic:input_x}
    a_{1 : n_x}^{(0)} := x_l,~
    b_{1 : n_x}^{(0)} := x_u,
\end{equation}
\begin{equation}\label{eq:interval_arithmetic:input_u}
a_{(n_x+1) : (n_x+n_u)}^{(0)} := u_l, ~
b_{(n_x+1) : (n_x+n_u)}^{(0)} := u_u,
\end{equation}
\begin{equation}\label{eq:interval_arithmetic:linear}
\begin{array}{c}
    [\hat{a}^{(i)}, \hat{b}^{(i)}] = {\tt lin\_layer}(a^{(i-1)}, b^{(i-1)}; W^{(i)}, B^{(i)}),
    \\
    i=1,\dots,\ell,
\end{array}
\end{equation}
\begin{equation}\label{eq:interval_arithmetic:relu}
\begin{array}{c}
    \bsq{a^{(i)},~  b^{(i)}} = \bsq{\max\{0, \hat{a}^{(i)}\}, ~\max\{0, \hat{b}^{(i)}\}},
    \\
    i=1,\dots,\ell-1,
\end{array}
\end{equation}
\end{subequations}
and function ${\tt lin\_layer}$ produces the hyperbox $[\hat{a}^{(i)}, \hat{b}^{(i)}]$ according to the following linear matrix operations,
\begin{multline}\label{eq:interval_arithmetic:linear:define}
{\tt lin\_layer}(a^{(i-1)}, b^{(i-1)}; W^{(i)}, B^{(i)}) :=
\\
\left\{
\begin{array}{l}
    \hat{a}_{j}^{(i)}=w^{(i)}_j S_i\left((w^{(i)}_j)^T\right)
    \left[\begin{array}{c}
        a^{(i-1)}\\
        b^{(i-1)}
    \end{array}\right]
    + B^{(i)}_j,
    \\
    \hat{b}_{j}^{(i)}=w^{(i)}_j S_i\left((w^{(i)}_j)^T\right)
    \left[\begin{array}{c}
        b^{(i-1)}\\
        a^{(i-1)}
    \end{array}\right]
    + B^{(i)}_j,
    \\
    j=1,\cdots,n_i.
\end{array}
\right.
\end{multline}
Here, the notation $a_{p:q}^{(0)}$ represents the sub-vector of $a^{(0)}$ containing elements between rows $p$ and $q$. The subscript $j$ in variables $\hat{a}_{j}^{(i)}$, $\hat{b}_{j}^{(i)}$ and $B_{j}^{(i)}$ denotes the $j$th element of the respective vectors. The row vector $w^{(i)}_j$ is the $j${th} row of the weight matrix $W^{(i)}$ in NNDS.
The functions $S_i:\mathbb{R}^{n_{i-1}}\rightarrow\mathbb{R}^{n_{i-1}\times2n_{i-1}}$ and $s: \mathbb{R}\rightarrow\mathbb{R}^{2n_{i-1}}$ are defined according to 
\begin{multline}\label{eq:interval_arithmetic:linear:define:switch}
    S_i
    \left(\left[\begin{array}{c}
        \vdots\\
        w_q\\
        \vdots
    \end{array}\right]\right)
    =
    \left[\begin{array}{c}
        \vdots\\
        \left(s(w_q)\right)^T\\
        \vdots
    \end{array}\right],
    \\
    s(w_q) :=
    \left\{\begin{array}{cc}
        \left[\begin{array}{c}
            e_q\\
            \zero_{n_{i-1}\times 1}
        \end{array}\right]
        & \text{if } w_q \geq 0 
        \\
        \left[\begin{array}{c}
            \zero_{n_{i-1}\times 1}\\
            e_q
        \end{array}\right]
        & \text{if } w_q < 0 
    \end{array}\right.,
\end{multline}
where $\zero_{n \times m}$ is a matrix filled entirely with zeros of size $n \times m$, and $e_q\in\mathbb{R}^{n_{i-1}}$ ($q= 1,\dots,n_{i-1}$) has 1 as the $q${th} element and zero anywhere else. 
Intuitively, in \eqref{eq:interval_arithmetic:linear:define}, the $q${th} row of matrix $S_i((w^{(i)}_j)^T)$ switches the lower and upper bounds of $z^{(i-1)}$, i.e., $a^{(i-1)}$ and $b^{(i-1)}$, if the $q${th} element in $w^{(i)}_j$ is negative, i.e., $w_q < 0$.

Eventually, the arithmetic operations in \eqref{eq:reachability:interval_arithmetic} and \eqref{eq:interval_arithmetic} determine the following set propagation through the NNDS:
\begin{equation}\label{eq:ab_chain}
    \begin{array}{c}
        \bsq{a^{(0)}, b^{(0)}}
        \to \bsq{\hat{a}^{(1)}, \hat{b}^{(1)}}
        \to \bsq{a^{(1)}, b^{(1)}}
        \\
        \to\dots\to 
        \\
        \bsq{\hat{a}^{(\ell-1)}, \hat{b}^{(\ell-1)}}
        \to \bsq{a^{(\ell-1)}, b^{(\ell-1)}}
        \to \bsq{\hat{a}^{(\ell)}, \hat{b}^{(\ell)}}.
    \end{array}
\end{equation}
In particular, \eqref{eq:interval_arithmetic:input_x} and \eqref{eq:interval_arithmetic:input_u} establish bounds on the NNDS input $z^{(0)}$ as defined in \eqref{eq:nn_structure:in}. These constraints ensure that $z^{(0)} = [x_k^T~u_k^T]^T$ lies within $\left[a^{(0)}, b^{(0)}\right]$, for all $x_k \in \X_k$ and $u_k \in \U_k$. 
From the $(i-1)$th to $i$th layer of the NNDS, set propagation is achieved using vectors $a^{(i-1)}$, $b^{(i-1)}$ (bounds on $z^{(i-1)}$) and $\hat{a}^{(i)}$, $\hat{b}^{(i)}$ (bounds on $\hat{z}^{(i)}$), ensuring that $z^{(i-1)} \in [a^{(i-1)}, b^{(i-1)}]$ and $\hat{z}^{(i)} \in [\hat{a}^{(i)}, \hat{b}^{(i)}]$.
In particular, the arithmetics \eqref{eq:interval_arithmetic:linear} and \eqref{eq:interval_arithmetic:linear:define} compute the set propagation through the fully connected layers defined by \eqref{eq:nn_structure:linear}, handling linear transformations. 
The propagation through the nonlinear ReLU function defined by \eqref{eq:nn_structure:relu} is derived using \eqref{eq:interval_arithmetic:relu}.  

\begin{remark}
The Reachability Analysis above is realized using fixed input hyperboxes and arithmetic operations.  
We note that the matrix operations defined in \eqref{eq:interval_arithmetic:linear} and \eqref{eq:interval_arithmetic:linear:define} are linear in the vectors $\hat{a}^{(i)}, \hat{b}^{(i)}, a^{(i-1)}, b^{(i-1)}$, whereas this property does not hold for the nonlinear $\max$ operations used in \eqref{eq:interval_arithmetic:relu}.  
In contrast, in Problem~\ref{problem:1:reduced}, a similar set propagation chain to \eqref{eq:ab_chain} is initiated with $[a^{(0)}, b^{(0)}]$, which depends on the unknown decision variable $u_k$, making this process decision-variable-dependent.  
Thus, the main focus in addressing Problem~\ref{problem:1:reduced} is the treatment of the nonlinear $\max$ operations.
\end{remark}


To address Problem~\ref{problem:1:reduced} using constrained optimization, the primary challenge lies in handling the decision-variable-dependent set propagation $\F(\I_i, \{u_k\})$ and the embedded nonlinear ${\tt ReLU}$ activation. To overcome this, we reformulate the set-valued mapping $\F(\I_i, \{u_k\})$, which depends on the decision variable $u_k$, as constraints within an optimization problem.  
Additionally, we enforce the set inclusion constraint $\F(\I_i, \{u_k\}) \subseteq \T$, where $\T$ represents the target set formed by a union of hyperboxes. Then, Problem~\ref{problem:1:reduced} can be reformulated as verifying the existence of a solution to the following Mixed-Integer Linearly Constrained Optimization Problem (MILC-OP):
\begin{subequations}\label{eq:synthesis}
\begin{equation}\label{eq:synthesis:obj}
    \argmin\limits_{\substack{
    u_k,~ a^{(0)},~ b^{(0)},~\hat{a}^{(\ell)},~ \hat{b}^{(\ell)},~ \underline{x}_{k+1},~ \overline{x}_{k+1},
    \\
    \hat{a}^{(i)},~ \hat{b}^{(i)},~a^{(i)},~ b^{(i)},~ \alpha^{(i)},~ \beta^{(i)},~ \gamma^{(i)},~ i=1,\dots,\ell-1,
    \\
    \phi^{(m)},~ \psi^{(m)},~ m=1,\dots,n_o
    }} (*)
\end{equation}
\text{subject to:}
\begin{equation}\label{eq:synthesis:input}
\begin{array}{c}
    a_{1 : n_x}^{(0)} = \underline{x}_k,~
    a_{(n_x+1) : (n_x+n_u)}^{(0)} = u_k \geq \underline{u},
    \\
    b_{1 : n_x}^{(0)} = \overline{x}_k, ~
    b_{(n_x+1) : (n_x+n_u)}^{(0)} = u_k \leq \overline{u},    
\end{array}
\end{equation}
\begin{equation}\label{eq:synthesis:linear}
\begin{array}{c}  
     [\hat{a}^{(i)}, \hat{b}^{(i)}] = {\tt lin\_layer}(a^{(i-1)}, b^{(i-1)}; W^{(i)}, B^{(i)}),
     \\
     \underline{\hat{z}}^{(i)}\leq\hat{a}^{(i)}\leq\hat{b}^{(i)}\leq\overline{\hat{z}}^{(i)},
     ~
    i=1,\dots,\ell,
\end{array}
\end{equation}
\begin{equation}\label{eq:synthesis:relu}
\begin{array}{c}  
    {\tt MILC\_ReLU}
    \brk{
    \begin{array}{c}
         a^{(i)}, b^{(i)}, \hat{a}^{(i)}, \hat{b}^{(i)}, 
         \\
         \alpha^{(i)}, \beta^{(i)}, \gamma^{(i)}
    \end{array}
    ;
    \underline{\hat{z}}^{(i)}, \overline{\hat{z}}^{(i)}
    }
    , 
    \\
    i=1,\dots,\ell-1,
\end{array}
\end{equation}
\begin{equation}\label{eq:synthesis:output}
    \underline{x}_{k+1}=\hat{a}^{(\ell)},
    ~
    \overline{x}_{k+1} = \hat{b}^{(\ell)},
\end{equation}
\begin{equation}\label{eq:synthesis:subset}
    {\tt MILC\_Inc}
    \brk{
    \hspace{-0.3em}
    \begin{array}{c}
         \underline{x}_{k+1}, \overline{x}_{k+1}, \\
         \{{\phi^{(m)}, \psi^{(m)}}\}_{m=1}^{n_o}
    \end{array}
    \hspace{-0.3em}
    ; 
    \hspace{-0.3em}
    \begin{array}{c}
         \underline{x}, \overline{x}, \X \backslash\T = 
         \\
         \cup_{m=1}^{n_o} \obsq{\underline{o}^{(m)}, \overline{o}^{(m)}}
    \end{array}
    }
\end{equation}
\end{subequations}
where operator ${\tt MILC\_ReLU}$ derives a collection of Mixed-Integer Linear Constraints (MILCs) encoding the decision-variable-depend uncertainty set propagation through nonlinear ${\tt ReLU}$ activations according to
\begin{multline}\label{eq:synthesis:relu:define}
{\tt MILC\_ReLU}
\brk{
\begin{array}{c}
     a^{(i)}, b^{(i)}, \hat{a}^{(i)}, \hat{b}^{(i)}, 
     \\
     \alpha^{(i)}, \beta^{(i)}, \gamma^{(i)}
\end{array}    
;
\underline{\hat{z}}^{(i)}, \overline{\hat{z}}^{(i)}
} 
:= 
\\
\left\{
\begin{array}{l}
    \alpha_j^{(i)}, \beta_j^{(i)}, \gamma_j^{(i)}\in\{0,1\},
    ~
    \alpha_j^{(i)}+ \beta_j^{(i)}+\gamma_j^{(i)} = 1,
    \\
    j=1,\cdots,n_i,
    \\
    a^{(i)} \geq \hat{a}^{(i)},
    ~        
    a^{(i)} \leq \diag(\overline{\hat{z}}^{(i)})\gamma^{(i)},
    \\
    a^{(i)} \leq \hat{a}^{(i)} - \diag(\underline{\hat{z}}^{(i)})(\alpha^{(i)}+\beta^{(i)}),
    \\
    b^{(i)} \geq \hat{b}^{(i)}, 
    ~
    b^{(i)} \leq \hat{b}^{(i)} - \diag(\underline{\hat{z}}^{(i)})\alpha^{(i)},
    \\
    b^{(i)} \leq \diag(\overline{\hat{z}}^{(i)})(\beta^{(i)}+\gamma^{(i)}),
    ~
    0 \leq a^{(i)} \leq b^{(i)},
\end{array}
\right.
\end{multline}
and another operator ${\tt MILC\_ReLU}$ generates a collection of MILCs to enforce the set inclusion relationship, i.e., $[\underline{x}_{k+1}, \overline{x}_{k+1}]\subseteq \T$, according to
\begin{multline}\label{eq:synthesis:subset:define}
{\tt MILC\_Inc}
\brk{
\begin{array}{c}
     \underline{x}_{k+1}, \overline{x}_{k+1}, \\
     \{{\phi^{(m)}, \psi^{(m)}}\}_{m=1}^{n_o}
\end{array}
\hspace{-0.3em}
; 
\hspace{-0.2em}
\begin{array}{c}
     \underline{x}, \overline{x}, \X \backslash\T = \X \cap 
     \\
     \brk{\cup_{m=1}^{n_o} ]\underline{o}^{(m)}, \overline{o}^{(m)}[}
\end{array}
}
\\
:=
\left\{
\begin{array}{l}
    \phi^{(m)}_j,~\psi^{(m)}_j\in\{0,1\},
    ~
    \phi^{(m)}_j + \psi^{(m)}_j\leq 1,
    \\
    j=1,\dots,n_x,
    \\
    \sum\limits_{j=1}^{n_x}\left(\phi^{(m)}_j + \psi^{(m)}_j\right)\geq 1,
    \\
    \overline{x}_{k+1} \leq \overline{x} + \diag\left(\underline{o}^{(m)}-\overline{x}\right)\phi^{(m)},
    \\
    \overline{x}_{k+1} \geq \underline{o}^{(m)} - \diag\left(\underline{o}^{(m)}-\underline{x}\right)\phi^{(m)},
    \\
    \underline{x}_{k+1} \geq \underline{x} + \diag\left(\overline{o}^{(m)}-\underline{x}\right)\psi^{(m)},
    \\
    \underline{x}_{k+1} \leq \overline{o}^{(m)} - \diag\left(\overline{o}^{(m)}-\overline{x}\right)\psi^{(m)},
    \\
    m=1,\dots,n_o.
\end{array}
\right.
\end{multline}
Here, the subscript $j$ in variables $\alpha_{j}^{(i)}$, $\beta_{j}^{(i)}$, $\gamma_{j}^{(i)}$, $\phi_{j}^{(m)}$, $\psi_{j}^{(m)}$, which are integer (binary) decision variables, denotes the $j$th element of the respective vectors. $\diag(x)\in\RR^{n\times n}$ yields a square matrix with elements of $x\in\RR^n$ on the diagonal and zero anywhere else.
Similar to the Reachability Analysis defined by \eqref{eq:reachability}, the constraints in \eqref{eq:synthesis} embed a similar set propagation process as in \eqref{eq:ab_chain}. However, unlike the previous case, here the hyperboxes $\{[\hat{a}^{(i)}, \hat{b}^{(i)}]\}_{i=1}^{\ell}$ and $\{[a^{(i)}, b^{(i)}]\}_{i=1}^{\ell-1}$ are decision variables, making the set propagation process in \eqref{eq:ab_chain} decision-variable-dependent.

In MILC-OP \eqref{eq:synthesis}, the constraints \eqref{eq:synthesis:input} establish bounds on the NNDS input $z^{(0)}$ as defined in \eqref{eq:nn_structure:in}. These constraints ensure that $z^{(0)} = [x_k^T~u_k^T]^T$ lies within $\left[a^{(0)}, b^{(0)}\right]$ for all $x_k \in \I_i = [\underline{x}_k, \overline{x}_k]$ (from Problem~\ref{problem:1:reduced}) and decision variable $u_k \in \U$, where $\U$ is defined in \eqref{eq:u_feasible_set}.  
Similar to \eqref{eq:interval_arithmetic:linear}, the linear constraints \eqref{eq:synthesis:linear} encode the set propagation through the fully connected layers defined by \eqref{eq:nn_structure:linear}, ensuring the correct handling of linear transformations.  
Different from \eqref{eq:interval_arithmetic:relu}, which involves nonlinear operations, the propagation through the nonlinear ${\tt ReLU}$ function defined by \eqref{eq:nn_structure:relu} is enforced using MILCs \eqref{eq:synthesis:relu} and \eqref{eq:synthesis:relu:define}, which is first developed in \cite{li2024system} and ensures an accurate modeling of the activation function's behavior.

In particular, since the values $z^{(i-1)}$ in the neurons are confined within a bounded hyperbox $[a^{(i-1)}, b^{(i-1)}]$, where $a^{(i-1)}, b^{(i-1)}$ are decision variables, the activation status of each ${\tt ReLU}$ function is uncertain. To address this uncertainty, we introduce integer variables $\alpha^{(i)}$, $\beta^{(i)}$, and $\gamma^{(i)}$ in constraints \eqref{eq:synthesis:relu} and \eqref{eq:synthesis:relu:define}. These variables encode the following possible activation statuses of the ReLU function, 
\begin{equation}\label{eq:activation_abg}
\left\{
\begin{array}{c}
    \hat{a}_{j}^{(i)}\leq \hat{b}_{j}^{(i)}\leq 0 \text{ if } \alpha_j^{(i)} = 1,
    \\
    \hat{a}_j^{(i)}\leq 0 \leq\hat{b}_j^{(i)} \text{ if } \beta_j^{(i)} = 1,
    \\
    0\leq\hat{a}_j^{(i)}\leq \hat{b}_j^{(i)} \text{ if } \gamma_j^{(i)} = 1.
\end{array}
\right.
\end{equation}
The constraints \eqref{eq:synthesis:linear}, when $i=\ell$, yield hyperbox $[\hat{a}^{(\ell)}, \hat{b}^{(\ell)}]$ and encode the set propagation through the final fully connected layer defined by \eqref{eq:nn_structure:out}. 
Finally, the hyperbox defined by the  decision variables $\underline{x}_{k+1}$ and $\overline{x}_{k+1}$ in \eqref{eq:synthesis:output} represents a superset of the decision-variable-dependent reachable set $\F(\I_i, \{u_k\})$, i.e., $\F(\I_i, \{u_k\}) \subseteq [\underline{x}_{k+1}, \overline{x}_{k+1}]$.

\begin{remark}\label{rmk:nn_bounds}
    Furthermore, using the bounded hyperboxes $\X$ in \eqref{eq:x_feasible_set} and $\U$ in \eqref{eq:u_feasible_set}, we can numerically derive lower and upper bounds $\underline{\hat{z}}^{(i)},\overline{\hat{z}}^{(i)}\in\mathbb{R}^{n_i},i=1,\dots,\ell$ on the neuron values $\hat{z}^{(i)}$ and the output $x_{k+1}$ using interval arithmetic, i.e., 
    \begin{equation}\label{eq:rmk_nn_bounds}
        \bcur{[\underline{\hat{z}}^{(i)}, \overline{\hat{z}}^{(i)}]}_{i=1}^{\ell}= {\tt reachable\_hyperboxes} (\X, \U) 
    \end{equation}
    such that $\hat{z}^{(i)}\in[\underline{\hat{z}}^{(i)},\overline{\hat{z}}^{(i)}]$ and $x_{k+1}\in[\underline{\hat{z}}^{(\ell)},\overline{\hat{z}}^{(\ell)}]$. These derived bounds are used to tighten the constraints \eqref{eq:synthesis:linear}, \eqref{eq:synthesis:relu} and limit the search region for the optimization solver.
\end{remark}

Lastly, the remaining constraints \eqref{eq:synthesis:subset} defined by  \eqref{eq:synthesis:subset:define} enforce the set inclusion condition $\F(\I_i, \{u_k\}) \subseteq \T$ by utilizing the established overbounds $[\underline{x}_{k+1}, \overline{x}_{k+1}]$ derived from the previous constraints \eqref{eq:synthesis:input}-\eqref{eq:synthesis:output}. Specifically, since $\F(\I_i, \{u_k\}) \subseteq [\underline{x}_{k+1}, \overline{x}_{k+1}]$, the inclusion condition is enforced by applying a stronger requirement $[\underline{x}_{k+1}, \overline{x}_{k+1}] \subseteq \T$.
Given the target set $\T = \tilde{\A}_i(\X_s)$, which is a finite union of closed hyperboxes as specified in Problem~\ref{problem:1:reduced}, and $\X$ in \eqref{eq:x_feasible_set}, which is itself a closed and bounded hyperbox, the target set can be expressed as $\T = \X \backslash \O$, where $\O = \bigcup_{m=1}^{n_o} \obsq{\underline{o}^{(m)}, \overline{o}^{(m)}}$ is a finite union of open hyperboxes. 
Eventually, this stronger inclusion $[\underline{x}_{k+1}, \overline{x}_{k+1}] \subseteq \T$ is equivalent to the condition of an empty intersection, specifically $\bsq{\underline{x}_{k+1}, \overline{x}_{k+1}} \cap \obsq{\underline{o}^{(m)}, \overline{o}^{(m)}} = \varnothing$ for all $m = 1, \dots, n_o$, which is encoded in \eqref{eq:synthesis:subset} and  \eqref{eq:synthesis:subset:define}.

Ultimately, we arrive at a CIS synthesis process formulated according to the following definition.

\begin{definition}[\textbf{CIS Synthesis Process}]
    In the sequel, the CIS synthesis process refers to the algorithmic framework comprising Algorithms~\ref{al:synthesisInvSet}, \ref{al:synthesisUnderQ} and Algorithm~\ref{al:isReturnable}, where line 3 and line 6 of Algorithm~\ref{al:isReturnable} are computed using \eqref{eq:reachability} and evaluated by solving the MILC-OP \eqref{eq:synthesis}, respectively.
\end{definition}

\begin{remark}\label{rmk:u_cis}
    As a byproduct of the CIS synthesis process, lines 6 and 7 in Algorithm~\ref{al:isReturnable} associate each verified hyperbox $\I_j^+$ with an admissible control $u_k$. These hyperboxes are subsequently merged into the $(i+1)$-Step Admissible Subset $\tilde{\A}_{i+1}(\X_s)$ in Algorithm~\ref{al:synthesisInvSet} through the set $\tilde{\R} = \bigcup_j \I_j^+$ in Algorithm~\ref{al:synthesisUnderQ}, producing the following control law, 
    \begin{subequations}
        \begin{multline}
            \pi \brk{x_k; \tilde{\A}_{i+1}(\X_s) } := u_k,
            \\
            \text{s.t. } 
            u_k \text{ solved MILC-OP \eqref{eq:synthesis} with } \T=\tilde{\A}_{i}(\X_s)
            \\
            \text{and } 
            \I_j = 
            \I \brk{x_k; \tilde{\A}_{i+1}(\X_s)} =
            [\underline{x}_k, \overline{x}_k],
        \end{multline}
        \begin{multline}
            \I \brk{x_k; \tilde{\A}_{i+1}(\X_s)} := \I_j^+,
            \\
            \text{ s.t. } x_k\in \I_j^+ \subseteq \cup_j \I_j^+ =  \tilde{\A}_{i+1}(\X_s) ).
        \end{multline}
    \end{subequations}
    This control law $\pi(\cdot~; \tilde{\A}_{i+1}(\X_s))$ renders $\tilde{\A}_{i+1}(\X_s)$ an $(i+1)$-Step Admissible Subset. Similarly, upon termination, $\pi(\cdot~; \C)$ renders $\C$ a CIS, which can further serve as a backup control plan in the event of a control system failure, providing an additional layer of operational reliability.
\end{remark}

\begin{remark}
    The total number of continuous decision variables is given by $6n_x + 3n_u + 4\sum_{i=1}^{\ell-1} n_i$, while the total number of integer or binary decision variables is $3\sum_{i=1}^{\ell-1} n_i + 2n_o$. Notably, both quantities scale linearly with the number of neurons in the NNDS, i.e., $\sum_{i=1}^{\ell-1} n_i$.
\end{remark}

\begin{remark}\label{rmk:inv:delete_constraints}
    To further enhance computational efficiency and reduce the number of decision variables, we can eliminate specific open hyperboxes $\obsq{\underline{o}^{(m)}, \overline{o}^{(m)}}$ from the union $\O$ when formulating the constraints \eqref{eq:synthesis:subset}. This deletion can be applied if it is determined that $\bar{\F}(\I_i, \U) \cap \bsq{\underline{o}^{(m)}, \overline{o}^{(m)}} = \varnothing$, and $\bar{\F}(\I_i, \U)$ is an over-approximated reachable set using \eqref{eq:reachability:Fbar}. 
\end{remark}

\begin{remark}
    The constraints in MILC-OP \eqref{eq:synthesis} are linear with integer decision variables.
    We also note that the optimization objective \eqref{eq:synthesis:obj} can be arbitrary, depending on the desired properties of the control solution. For instance, minimizing the $L_1$ norm of the control input, $\argmin \|u_k\|_1$, can be rewritten as a linear objective by introducing slack variables $\lambda = [\lambda_1~\lambda_2~\cdots~\lambda_{n_u}]^T$ and the following linear inequalities: $-\lambda \leq u_{k} \leq \lambda$. 
    This reformulation allows the objective to be expressed as $\argmin \sum_{i=1}^{n_u} \lambda_i$, which is a linear objective. Then, the MILC-OP \eqref{eq:synthesis} is a Mixed-Integer Linear Program (MILP). Alternatively, with Quadratic objective functions, \eqref{eq:synthesis} becomes a Mixed-Integer Quadratic Program (MIQP).
\end{remark}
\section{Theoretical Properties}\label{sec:methodSynThm}
In this section, we establish the theoretical properties of the proposed Algorithms~\ref{al:synthesisInvSet}, \ref{al:synthesisUnderQ}, and \ref{al:isReturnable}. Additionally, we analyze the properties of the Reachability Analysis \eqref{eq:reachability} in Sec.~\ref{method:thmProp:reach} and the MILC-OP \eqref{eq:synthesis} in Sec.~\ref{method:thmProp:milcop}, which are used to address lines 3 and 6 in Algorithm~\ref{al:isReturnable}, respectively. As a result, the CIS synthesis process yields a CIS within a finite number of iterations, as detailed in Sec.~\ref{method:thmProp:cis_finite}. 

\subsection{Reachability Analysis}\label{method:thmProp:reach}
To begin with, the properties of the Reachability Analysis \eqref{eq:reachability} are summarized in the following results.
\begin{lemma}\label{lemma:interval_arithmetic:linear}
    If the post-activation vector in the $(i-1)${th} layer, $z^{(i-1)}$, belongs to the hyperbox $[a^{(i-1)},~b^{(i-1)}]$, then the pre-activation vector in the $i${th} layer is given by $\hat{z}^{(i)} = W^{(i)}z^{(i-1)} + b^{(i)}$ and lies within the hyperbox $[\hat{a}^{(i)},~\hat{b}^{(i)}]$. The bounds $\hat{a}^{(i)}$ and $\hat{b}^{(i)}$ are derived using the function ${\tt lin\_layer}$ according to \eqref{eq:interval_arithmetic:linear:define}.
\end{lemma}

\begin{proof}
Assuming $z^{(i-1)}\in \left[a^{(i-1)},\;b^{(i-1)}\right]$, the $j${th} element of $\hat{z}^{(i)}$ can be expressed as $\hat{z}^{(i)}_j=w^{(i)}_jz^{(i-1)}+B^{(i)}_j$ where $w^{(i)}_j=[w_1\dots w_q\dots w_{n_{i-1}}]$. The following inequalities hold, $L\leq \hat{z}^{(i)}_j - B^{(i)}_j \leq U $, where $L = \sum_{q=1}^{n_{i-1}} \left( \IversonBracket{w_q\geq0}\cdot(w_q a^{(i-1)}_q) + \IversonBracket{w_q<0}\cdot(w_q b^{(i-1)}_q)\right)$, $U=\sum_{q=1}^{n_{i-1}} \left( \IversonBracket{w_q\geq0}\cdot(w_q b^{(i-1)}_q) + \IversonBracket{w_q<0}\cdot(w_q a^{(i-1)}_q)\right)$, and the Iverson bracket $\IversonBracket{\cdot}$ equals 1 if the condition inside is true and 0 otherwise. 
Meanwhile, due to the definition of function $S_i$ in \eqref{eq:interval_arithmetic:linear:define:switch}, the equations \eqref{eq:interval_arithmetic:linear} and \eqref{eq:interval_arithmetic:linear:define} simplify to $\hat{a}^{(i)}_j - B^{(i)}_j = L$ and $\hat{b}^{(i)}_j - B^{(i)}_j = U$, which implies $\hat{z}^{(i)}_j\in[\hat{a}^{(i)}_j,\;\hat{b}^{(i)}_j]$ for all $j=1,\dots,n_i$.
\end{proof}

\begin{theorem}\label{thm:reachability}
    Given arbitrary hyperboxes $\X_k=[x_l, x_u]$, $\U_k=[u_l, u_u]$, the Reachability Analysis $\bar{\F}$ defined according to \eqref{eq:reachability} and \eqref{eq:interval_arithmetic} generates a hyperbox $\bar{\F}(\X_k, \U_k):= [\hat{a}^{(\ell)}, \hat{b}^{(\ell)}]$ such that $\F(\X_k, \U_k)\subseteq \bar{\F}(\X_k, \U_k)$. 
\end{theorem}

\begin{proof}
    By construction in \eqref{eq:interval_arithmetic:input_x} and \eqref{eq:interval_arithmetic:input_u}, the input to the NNDS, $z^{(0)}= [x_k^T ~ u_k^T]^T$, lies within the hyperbox $\X_k \times \U_k = [a^{(0)}, b^{(0)}]$, where $\times$ is the Cartesian product. Consequently, the input $z_0$ satisfies the following condition, $z^{(0)}=[x_k^T\; u_k^T]^T\in[a^{(0)},\;b^{(0)}]$.
    For the set propagation through NNDS hidden layers encoded in \eqref{eq:interval_arithmetic:linear} and \eqref{eq:interval_arithmetic:relu}, we first prove the following statement: if $z^{(i-1)}\in [a^{(i-1)},~b^{(i-1)}]$, then $\hat{z}^{(i)} = W^{(i)}z^{(i-1)} +b^{(i)}\in [\hat{a}^{(i)},~\hat{b}^{(i)}]$, and $z^{(i)} = \max\{0,~\hat{z}^{(i)}\}\in [a^{(i)},~b^{(i)}]$. Then, using the initial condition $z^{(0)}\in[a^{(0)},~b^{(0)}]$ above, we can inductively establish the result for $i=1,\dots,\ell-1$. 
    
    Indeed, if $z^{(i-1)} \in [a^{(i-1)},\;b^{(i-1)}]$, it follows from Lemma~\ref{lemma:interval_arithmetic:linear} that $\hat{z}^{(i)} \in [\hat{a}^{(i)},~\hat{b}^{(i)}]$. Subsequently, by definition in \eqref{eq:interval_arithmetic:relu}, it is evident that $z^{(i)} \in [a^{(i)},~b^{(i)}]$.  
    Then, we can show that $\hat{z}^{(i)} \in [\hat{a}^{(i)},\;\hat{b}^{(i)}],\;z^{(i)} \in [a^{(i)},\;b^{(i)}]$, for all $i=1,\dots,\ell-1$, inductively with $z^{(0)}\in[a^{(0)},~b^{(0)}]$. Identical to the proof above, given $z_{\ell-1} \in [a_{\ell-1},\;b_{\ell-1}]$ established above, and constraints \eqref{eq:interval_arithmetic:linear} at $i=\ell$ encoding the final output layer in \eqref{eq:nn_structure:out}, we can demonstrate that $x_{k+1} = f(x_k,u_k)\in [\hat{a}^{(\ell)}, \hat{b}^{(\ell)}]$ for all $x_k\in\X_k$ and $u_k\in \U_k$, which implies $\F(\X_k,\U_k)\subseteq [\hat{a}^{(\ell)}, \hat{b}^{(\ell)}]$    
\end{proof}

Eventually, the result below follows from Theorem~\ref{thm:reachability} and provides a solution to the Reachability Analysis in Problem~\ref{problem:1:reach}.

\begin{corollary}\label{coro:reachability}
    Given hyperboxes $\I_i = [\underline{x}_k, \overline{x}_k]$ and $\U = [\underline{u}, \overline{u}]$ in Problem~\ref{problem:1:reach}, the hyperbox $\bar{\F}(\X_k, \U_k)$ contains $\F(\I_i, \U)$, i.e., $\F(\I_i, \U)\subseteq \bar{\F}(\I_i, \U)$. 
\end{corollary}

\subsection{MILC-OP and Decision-Variable-Dependent Set Propagation}\label{method:thmProp:milcop}
Subsequently, we present the properties of the MILC-OP \eqref{eq:synthesis}, which addresses Problem~\ref{problem:1:reduced}, and establish its connection to the Reachability Analysis \eqref{eq:reachability}. One of the key components in the MILC-OP is the embedding of the decision-variable-dependent set propagation through the nonlinear ${\tt ReLU}$ function, which is implemented using the ${\tt MILC\_ReLU}$ constraints and the following results.

\begin{lemma}\label{lemma:synthesis:relu:equivalent}    
    Given the condition $\underline{\hat{z}}^{(i)}\leq\hat{a}^{(i)}\leq\hat{b}^{(i)}\leq\overline{\hat{z}}^{(i)}$ holds where $\underline{\hat{z}}^{(i)}, \overline{\hat{z}}^{(i)}$ are computed according to Remark~\ref{rmk:nn_bounds}, the ${\tt MILC\_ReLU}$ constraints \eqref{eq:synthesis:relu:define} is equivalent to \eqref{eq:interval_arithmetic:relu} for individual $i=1,\dots,\ell-1$, i.e., $[a^{(i)}, b^{(i)}]=[\max\{0, \hat{a}^{(i)}\},  \max\{0, \hat{b}^{(i)}\}]$ if and only if $a^{(i)}, b^{(i)}, \hat{a}^{(i)}, \hat{b}^{(i)}$ satisfy the ${\tt MILC\_ReLU}$ constrains in \eqref{eq:synthesis:relu:define}.
\end{lemma}

\begin{proof} 
    We note that the constraints \eqref{eq:synthesis:relu:define} ensure that exactly one of the integer variables $\alpha^{(i)}_j$, $\beta^{(i)}_j$, or $\gamma^{(i)}_j$ is equal to 1. This condition determines the uncertain activation status of the ${\tt ReLU}$ functions as detailed in \eqref{eq:activation_abg}. Then, we can establish the following categorical equivalency for constraint \eqref{eq:synthesis:relu:define}. 

    1.~If $\alpha^{(i)}_j=1$, then $\beta^{(i)}_j=\gamma^{(i)}_j=0$, and the following equivalency holds: \eqref{eq:synthesis:relu:define} hold 
    $\iff$ 
    $0 \leq a^{(i)}_j \leq b^{(i)}_j$, $\hat{a}^{(i)}_j \leq a^{(i)}_j \leq 0$, $\hat{b}^{(i)}_j \leq b^{(i)}_j \leq 0$, $a^{(i)}_j \leq \hat{a}^{(i)}_j -\underline{\hat{z}}^{(i)}_j$, $b^{(i)}_j \leq \hat{b}^{(i)}_j - \underline{\hat{z}}^{(i)}_j$
    $\iff$ 
    $a^{(i)}_j=b^{(i)}_j=0$, $\hat{a}^{(i)}_j\leq 0$, $\hat{b}^{(i)}_j\leq 0$, $a^{(i)}_j \leq \hat{a}^{(i)}_j -\underline{\hat{z}}^{(i)}_j$, $b^{(i)}_j \leq \hat{b}^{(i)}_j - \underline{\hat{z}}^{(i)}_j$
    $\iff$
    $[a^{(i)}_j, b^{(i)}_j]=[\max\{0, \hat{a}^{(i)}_j\}$, $\max\{0, \hat{b}^{(i)}_j\}] = [0,0]$, $a^{(i)}_j \leq \hat{a}^{(i)}_j -\underline{\hat{z}}^{(i)}_j$, $b^{(i)}_j \leq \hat{b}^{(i)}_j - \underline{\hat{z}}^{(i)}_j$ 
    where the last two inequalities, i.e, $a^{(i)}_j \leq \hat{a}^{(i)}_j -\underline{\hat{z}}^{(i)}_j$ and $b^{(i)}_j \leq \hat{b}^{(i)}_j - \underline{\hat{z}}^{(i)}_j$, hold trivially since $\underline{\hat{z}}^{(i)}_j\leq\hat{a}^{(i)}_j\leq \hat{b}^{(i)}_j\leq \overline{\hat{z}}^{(i)}_j$.

    2.~If $\beta^{(i)}_j=1$, then $\alpha^{(i)}_j=\gamma^{(i)}_j=0$, and the following equivalency holds: \eqref{eq:synthesis:relu:define} hold 
    $\iff$
    $0 \leq a^{(i)}_j \leq b^{(i)}_j$, $\hat{a}^{(i)}_j \leq a^{(i)}_j \leq 0$, $\hat{b}^{(i)}_j \leq b^{(i)}_j \leq \hat{b}^{(i)}_j$, $a^{(i)}_j \leq \hat{a}^{(i)}_j -\underline{\hat{z}}^{(i)}_j$, $b^{(i)}_j \leq \overline{\hat{z}}^{(i)}_j$
    $\iff$
    $a^{(i)}_j=0$, $\hat{a}^{(i)}_j\leq 0$, $b^{(i)}_j=\hat{b}^{(i)}_j$, $\hat{b}^{(i)}_j\geq 0$, $a^{(i)}_j \leq \hat{a}^{(i)}_j -\underline{\hat{z}}^{(i)}_j$, $b^{(i)}_j \leq \overline{\hat{z}}^{(i)}_j$
    $\iff$
    $[a^{(i)}_j, b^{(i)}_j]=[\max\{0, \hat{a}^{(i)}_j\}$, $\max\{0, \hat{b}^{(i)}_j\}] = [0,\hat{b}^{(i)}_j]$, $a^{(i)}_j \leq \hat{a}^{(i)}_j -\underline{\hat{z}}^{(i)}_j$, $b^{(i)}_j \leq \overline{\hat{z}}^{(i)}_j$
    where the last two inequalities, i.e, $a^{(i)}_j \leq \hat{a}^{(i)}_j -\underline{\hat{z}}^{(i)}_j$ and $b^{(i)}_j \leq \overline{\hat{z}}^{(i)}_j$, hold trivially since $\underline{\hat{z}}^{(i)}_j\leq\hat{a}^{(i)}_j\leq \hat{b}^{(i)}_j\leq \overline{\hat{z}}^{(i)}_j$.

    3.~If $\gamma^{(i)}_j=1$, then $\alpha^{(i)}_j=\beta^{(i)}_j=0$, and the following equivalency holds: \eqref{eq:synthesis:relu:define} hold 
    $\iff$
    $0 \leq a^{(i)}_j \leq b^{(i)}_j$, $\hat{a}^{(i)}_j \leq a^{(i)}_j \leq \hat{a}^{(i)}_j$, $\hat{b}^{(i)}_j \leq b^{(i)}_j \leq \hat{b}^{(i)}_j$, $a^{(i)}_j \leq \overline{\hat{z}}^{(i)}_j$, $b^{(i)}_j \leq \overline{\hat{z}}^{(i)}_j$
    $\iff$
    $a^{(i)}_j=\hat{a}^{(i)}_j$, $\hat{a}^{(i)}_j\geq 0$, $b^{(i)}_j=\hat{b}^{(i)}_j$, $\hat{b}^{(i)}_j\geq 0$, $a^{(i)}_j \leq \overline{\hat{z}}^{(i)}_j$, $b^{(i)}_j \leq \overline{\hat{z}}^{(i)}_j$
    $\iff$
    $[a^{(i)}_j, b^{(i)}_j]=[\max\{0, \hat{a}^{(i)}_j\},  \max\{0, \hat{b}^{(i)}_j\}] = [\hat{a}^{(i)}_j,\hat{b}^{(i)}_j]$, $a^{(i)}_j \leq \overline{\hat{z}}^{(i)}_j$, $b^{(i)}_j \leq \overline{\hat{z}}^{(i)}_j$,
    where the last two inequalities, i.e, $a^{(i)}_j \leq \overline{\hat{z}}^{(i)}_j$ and $b^{(i)}_j \leq \overline{\hat{z}}^{(i)}_j$, hold trivially since $\underline{\hat{z}}^{(i)}_j\leq\hat{a}^{(i)}_j\leq \hat{b}^{(i)}_j\leq \overline{\hat{z}}^{(i)}_j$.
    
    With the categorical equivalencies and for arbitrary $j=1,\dots,n_i$, we have proven that $a^{(i)}_j, b^{(i)}_j, \hat{a}^{(i)}_j, \hat{b}^{(i)}_j$ satisfy \eqref{eq:synthesis:relu:define} if and only if $[a^{(i)}_j, b^{(i)}_j]=[\max\{0, \hat{a}^{(i)}_j\},  \max\{0, \hat{b}^{(i)}_j\}]$.
\end{proof}

Then, the result below follows directly from Lemma~\ref{lemma:synthesis:relu:equivalent}.

\begin{corollary}\label{coro:synthesis:relu:contain}
    In the $i${th} layer, if the pre-activation vector in the $i${th} layer, $\hat{z}^{(i)}$, belongs to the hyperbox $[\hat{a}^{(i)},~\hat{b}^{(i)}]$,   then the post-activation vector $z^{(i)}=\max\{0, \hat{z}^{(i)}\}$ lies within the hyperbox $[a^{(i)},~b^{(i)}]$ if and only if $a^{(i)}, b^{(i)}, \hat{a}^{(i)}, \hat{b}^{(i)}$ satisfy the ${\tt MILC\_ReLU}$ constrains in \eqref{eq:synthesis:relu:define}.
\end{corollary}

Meanwhile, the following results establish a connection between the MILC-OP~\eqref{eq:synthesis} and the Reachability Analysis \eqref{eq:reachability}, demonstrating the constraints in \eqref{eq:synthesis} encode a set propagation process that depends on the control decision variable $u_k$.

\begin{theorem}\label{thm:over_approx}
Given a hyperbox $\I_i$ and a union of hyperboxes $\T$ as described in Problem~\ref{problem:1:reduced}, if a solution to the MILC-OP \eqref{eq:synthesis} exists, the values of the decision variables $u_k\in\U$, $\underline{x}_{k+1}$, and $\overline{x}_{k+1}$ produce the hyperbox $\bar{\F}(\I_i, \{u_k\})$, such that $\F(\I_i, \{u_k\}) \subseteq \bar{\F}(\I_i, \{u_k\}) =  [\underline{x}_{k+1}, \overline{x}_{k+1}] $.
\end{theorem}

\begin{proof}
By construction in constraints \eqref{eq:synthesis:input}, it follows that $u_k\in\U$ is admissible. The input to the NNDS, $z^{(0)}= [x_k^T ~ u_k^T]^T$, lies within the hyperbox $\I_i \times \{u_k\} = [a^{(0)}, b^{(0)}]$, where $\I_i = [\underline{x}_k, \overline{x}_k]$ as specified in Problem~\ref{problem:1:reduced}. Consequently, the input $z_0$ satisfies the following condition, $z^{(0)}=[x_k^T\; u_k^T]^T\in[a^{(0)},\;b^{(0)}]$. 
Afterward, for set propagation in the NNDS hidden layer, the linear layers in \eqref{eq:synthesis:linear} are identical to \eqref{eq:interval_arithmetic:linear} while the ${\tt MILC\_ReLU}$ in \eqref{eq:synthesis:relu} is equivalent to \eqref{eq:interval_arithmetic:relu} as established by Lemma~\ref{lemma:synthesis:relu:equivalent}. Therefore, we know that $\bar{\F}(\I_i, \{u_k\}) =  [\underline{x}_{k+1}, \overline{x}_{k+1}]$ given $z^{(0)}\in[a^{(0)},\;b^{(0)}]$. Subsequently, we have $\F(\I_i, \{u_k\}) \subseteq \bar{\F}(\I_i, \{u_k\})$ due to Theorem~\ref{thm:reachability}.
\end{proof}

\begin{remark}
    The constraints \eqref{eq:synthesis:input}-\eqref{eq:synthesis:output} encode a superset $\bar{\F}(\I_i, \{u_k\})$ of the actual decision-variable-dependent set $\F(\I_i, \{u_k\})$. They estimate the reachable sets containing neuron values in each layer, $\hat{z}^{(i)}$, $i = 1, \dots, \ell-1$, using hyperboxes. However, the linear layers described in \eqref{eq:nn_structure:linear}, which consist of linear transformations, do not necessarily preserve the hyperbox shape. 
    This over-approximation disregards the dependencies between variables in the hidden layers, as also noted in \cite{dutta2018output}. While hyperboxes offer convenient set operations and enable decision-variable-dependent set propagation within \eqref{eq:synthesis} using MILCs, ensuring computational efficiency with well-established integer solvers \cite{gurobi}, future investigations will explore replacing hyperboxes with polytopes to better capture dependencies between variables in the hidden layers.
\end{remark}

Lastly, the MILC-OP \eqref{eq:synthesis} also employs constraints to enforce the set inclusion condition, i.e., $[\underline{x}_{k+1}, \overline{x}_{k+1}] \subseteq \T$. The results related to this inclusion are discussed below and, together with the theoretical results established above, contribute to solving Problem~\ref{problem:1:reduced}.

\begin{lemma}\label{lemma:synthesis:subset}
    Given a union of hyperboxes $\T$ as described in Problem~\ref{problem:1:reduced}, the hyperbox $[\underline{x}_{k+1}, \overline{x}_{k+1}]\subseteq \X$ is a subset of $\T$, i.e., $[\underline{x}_{k+1}, \overline{x}_{k+1}]\subseteq \T$, if and only if $\underline{x}_{k+1}$,  $ \overline{x}_{k+1}$ and $\T$ satisfy the ${\tt MILC\_Inc}$ constraints \eqref{eq:synthesis:subset:define}.
\end{lemma}

\begin{proof}
    Again, given the target set $\T = \tilde{\A}_i(\X_s)$, which is a finite union of closed hyperboxes as specified in Problem~\ref{problem:1:reduced}, and $\X$ in \eqref{eq:x_feasible_set}, which is itself a closed and bounded hyperbox, the target set can be expressed as $\T = \X \backslash \brk{\cup_{m=1}^{n_o} ]\underline{o}^{(m)}, \overline{o}^{(m)}[}$.
    Meanwhile, the integer variables $\phi^{(m)}$ and $\psi^{(m)}$ reduce the constraints \eqref{eq:synthesis:subset:define} to cases and have the following implications
    \begin{equation}\label{eq:phi_psi_implication}
        \hspace{-1em}
        \left\{
        \begin{array}{c}
            \underline{x}_{k+1,j}\leq \overline{x}_{k+1,j} \leq \underline{o}^{(m)}_j \quad \text{if } \phi^{(m)}_j = 1,~ \psi^{(m)}_j = 0,
            \\
            \overline{x}_{k+1,j}\geq \underline{x}_{k+1,j} \geq \overline{o}^{(m)}_j \quad \text{if } \phi^{(m)}_j = 0,~ \psi^{(m)}_j = 1,
            \\
            \begin{array}{l}
                \underline{o}^{(m)}_j\leq \overline{x}_{k+1,j}\leq\overline{x},
                \\
                \quad \underline{x}\leq \underline{x}_{k+1,j}\leq\overline{o}^{(m)}_j
            \end{array} 
            \quad \text{if } \phi^{(m)}_j =\psi^{(m)}_j = 0,
        \end{array}
        \right.
    \end{equation} 
    where the subscripts $j$ denote the $j${th} elements in the respective vectors and $\phi^{(m)}_j =\psi^{(m)}_j = 1$ is not possible due to the constraint $\phi^{(m)}_j + \psi^{(m)}_j\leq 1$.
    The following equivalency can be established to prove the results:
    \begin{multline*}
        [\underline{x}_{k+1}, \overline{x}_{k+1}]\subseteq \T, \text{ where } \T = \X \backslash \brk{\cup_{m=1}^{n_o} ]\underline{o}^{(m)}, \overline{o}^{(m)}[}
        \\
        \iff 
        ~
        \forall m = 1, \dots, n_o,~ [\underline{x}_{k+1},\overline{x}_{k+1}] ~\cap~ ]\underline{o}^{(m)}, \overline{o}^{(m)}[ = \varnothing
        \\
        \iff 
        ~
        \forall m = 1, \dots, n_o,~\exists j\in \{1,\dots, n_x\}, 
        \text{ s.t. }
        \\
        \underline{x}_{k+1,j}\leq \overline{x}_{k+1,j} \leq \underline{o}^{(m)}_j
        \text{ or }
        \overline{x}_{k+1,j}\geq \underline{x}_{k+1,j} \geq \overline{o}^{(m)}_j,  
        \\
        \iff
        ~
        \forall m = 1, \dots, n_o,~\exists j\in \{1,\dots, n_x\},
        \text{ s.t. }
        \\
        \phi^{(m)}_j=1 \text{ or } \psi^{(m)}_j=1
        \text{ with implication in \eqref{eq:phi_psi_implication} }
        \\
        \iff
        ~
        \forall m = 1, \dots, n_o,
        ~
        \Sigma_{j=1}^{n_x}\left(\phi^{(m)}_j + \psi^{(m)}_j\right)\geq 1,   
        \\
        \text{ and }
        \forall j=1,\dots,n_x,
        ~
        \phi^{(m)}_j,~\psi^{(m)}_j\in\{0,1\},
        ~
        \phi^{(m)}_j + \psi^{(m)}_j\leq 1,
        \\
        \text{ with implication in \eqref{eq:phi_psi_implication} }
        ~
        \iff
        \\
        \text{$\underline{x}_{k+1}$,  $ \overline{x}_{k+1}$ and $\T$ satisfy the ${\tt MILC\_Inc}$ constraints \eqref{eq:synthesis:subset:define}}
    \end{multline*}
    Here, the first equivalency is due to $\T \cap \brk{\cup_{m=1}^{n_o} ]\underline{o}^{(m)}, \overline{o}^{(m)}[} = \varnothing$ and $[\underline{x}_{k+1}, \overline{x}_{k+1}] \subseteq\X \subseteq \brk{\T \cup \brk{\cup_{m=1}^{n_o} ]\underline{o}^{(m)}, \overline{o}^{(m)}[}}$.
\end{proof}

\begin{theorem}\label{thm:one_step_return}
Given a hyperbox $\I_i$ and a union of hyperboxes $\T$ as described in Problem~\ref{problem:1:reduced}, if a solution to the MILC-OP~\eqref{eq:synthesis} exists, the value of control $u_k$ is admissible and regulates the state $x_k$ to the target set $\T$ in one step for all $x_k \in \I_i$, i.e., $u_k \in \U$ and $\F(\I_i, \{u_k\})\subseteq \bar{\F}(\I_i, \{u_k\})\subseteq \T$.
\end{theorem}

\begin{proof}
    Theorem~\ref{thm:over_approx} demonstrates that $\F(\I_i, \{u_k\}) \subseteq \bar{\F}(\I_i, \{u_k\}) =  [\underline{x}_{k+1}, \overline{x}_{k+1}] $ while the Lemma~\ref{lemma:synthesis:subset} establishes $[\underline{x}_{k+1}, \overline{x}_{k+1}]\subseteq \T$.
\end{proof}

As a result, in Algorithm~\ref{al:isReturnable}, we can verify the \textbf{if} statement in line 6 by examining the existence of a solution to MILC-OP \eqref{eq:synthesis} using Theorem~\ref{thm:one_step_return}, which addresses Problem~\ref{problem:1:reduced}. While MILC-OP \eqref{eq:synthesis} consists of linear constraints with integer variables, the over-approximation, i.e. $\F(\I_i, \{u_k\})\subseteq \bar{\F}(\I_i, \{u_k\})$, may lead to cases where MILC-OP \eqref{eq:synthesis} lacks a solution, meaning that $\nexists u_k \in \U$ such that $\bar{\F}(\I_i, \{u_k\}) \subseteq \T$, even though there exists $u_k \in \U$ such that $\F(\I_i, \{u_k\}) \subseteq \T$. This limitation motivates the partitioning of hyperboxes described in line 9 of Algorithm~\ref{al:synthesisUnderQ}.

\begin{remark}
    Based on Theory \ref{thm:over_approx} and \ref{thm:one_step_return}, if $\bar{\F}(\I_i, \U)\cap\T=\varnothing$ in line 3 of Algorithm~\ref{al:isReturnable}, this implies $\nexists u_k\in\U$ such that $\bar{\F}(\I_i, \{u_k\}) \subseteq \T$, subsequently, indicates the MILC-OP \eqref{eq:synthesis} has no solution. 
    It preemptively verifies the \textbf{if} statements in line 6 of Algorithm~\ref{al:isReturnable} false, i.e., $\nexists u_k\in\U$ such that $\F(\I_i, \{u_k\})\subseteq \T$ since $\F(\I_i, \{u_k\})\subseteq \bar{\F}(\I_i, \{u_k\})$.
\end{remark}

\subsection{CIS Synthesis Process and Finite Determination}\label{method:thmProp:cis_finite}

Ultimately, with Problem~\ref{problem:1:reach} and Problem~\ref{problem:1:reduced} addressed above, we provide the following theoretical results proving that the CIS synthesis process generates a CIS $\C$ while ensuring termination in a finite number of iterations.

\begin{lemma}\label{lemma:admissible}
    Given $\X$, $\U$, $\X_s$ obeying Assumption~\ref{assume:x_u_box}, and the SSQ process, consider the CIS synthesis process. If $\tilde{\A}_i(\X_s) \subseteq \A_i(\X_s)$ is an $i$-Step Admissible Subset and $\tilde{\A}_i(\X_s)\in \P_{\cup}(\X_\Delta)$, then $\tilde{\A}_{i+1}(\X_s) \subseteq \A_{i+1}(\X_s)$ is an $(i+1)$-Step Admissible Subset and $\tilde{\A}_{i+1}(\X_s)\in \P_{\cup}(\X_\Delta)$.
\end{lemma}

\begin{proof}
    Firstly, given that $\tilde{\A}_i(\X_s) \in \P_{\cup}(\X_\Delta)$, it follows that $\tilde{\A}_{i+1}(\X_s)=\Q(\tilde{\A}_i(\X_s), \tilde{\A}_{i}(\X_s)) = \tilde{\R} \in \P_{\cup}(\X_\Delta)$. The set $\tilde{\R}$ is computed using Algorithm~\ref{al:synthesisUnderQ} where the ${\tt partition\_hyperbox}$ procedure operates within $\P_{\cup}(\X_\Delta)$ to ensure that $\tilde{\R} = \cup_{j} \I_j^+$ and $\I_j^+ \in \P_{\cup}(\X_\Delta)$.    
    Secondly, given $\tilde{\A}_i(\X_s) \subseteq \A_i(\X_s)$ is an $i$-Step Admissible Subset, then for all initial states $x_0 \in \tilde{\A}_{i+1}(\X_s) =\Q(\tilde{\A}_i(\X_s), \tilde{\A}_{i}(\X_s))$, there must exist a hyperbox $\I_j^+ = \I (x_0; \tilde{\A}_{i+1}(\X_s))$ and a control $u_0 = u (x_0; \tilde{\A}_{i+1}(\X_s) ) \in \mathcal{U}$ defined according to Remark~\ref{rmk:u_cis} satisfying $x_1 = f(x_0, u_0) \in \F(\I_j^+, \{u_0\}) \subseteq \T$, where $\T = \tilde{\A}_i(\X_s)$, according to Theorem~\ref{thm:one_step_return}.
    Subsequently, since $\tilde{\A}_i(\X_s) \subseteq \A_i(\X_s)$, Definition~\ref{def:i_th_adm_subset} ensures the existence of a control sequence $\{u_k\}_{k=1}^{i} \subseteq \mathcal{U}$ such that $x_{k+1} = f(x_k, u_k) \in \mathcal{X}_s$ for all $k = 1, \dots, i$. As a result, the control sequence $\{u_k\}_{k=0}^{i}$ keeps the system state trajectory $\{x_k\}_{k=1}^{i+1}$ within $\X_s$, which implies that $\tilde{\A}_{i+1}(\X_s) \subseteq \A_{i+1}(\X_s)$.
\end{proof}

\begin{lemma}\label{lemma:admissible_nested}
    Given $\X$, $\U$, the safe set $\X_s$ obeying Assumption~\ref{assume:x_u_box}, and the SSQ process, the CIS synthesis process generates a nested sequence of $i$-Step Admissible Sets of $\X_s$, i.e., $\tilde{\A}_i(\X_s)\in \P_{\cup}(\X_{\Delta})$, $\tilde{\A}_i(\X_s) \subseteq \A_{i}(\X_s)$ for all $i \in \ZZ$ and $\tilde{\A}_j(\X_s) \subseteq \tilde{\A}_i(\X_s)\subseteq\X_s$ for all $i < j$.
\end{lemma}

\begin{proof}
    As established in Lemma~\ref{lemma:admissible}, if $\tilde{\A}_i(\X_s) \subseteq \A_i(\X_s)$ and $\tilde{\A}_i(\X_s)\in \P_{\cup}(\X_\Delta)$, then $\tilde{\A}_{i+1}(\X_s)\subseteq \A_{i+1}(\X_s)$ and $\tilde{\A}_{i+1}(\X_s)\in \P_{\cup}(\X_\Delta)$. Meanwhile, by construction, we have $\tilde{\A}_{i+1}(\X_s) = \Q(\tilde{\A}_{i}(\X_s),\tilde{\A}_{i}(\X_s))\subseteq \tilde{\A}_{i}(\X_s)$.
    Then, with $\tilde{\A}_0(\X_s) = \X_s^{\Delta} \subseteq \A_0(\X_s) = \X_s$ that holds trivially and $\X_s^{\Delta}\in \P_{\cup}(\X_\Delta)$, we can use induction to prove that $\tilde{\A}_i(\X_s)\in \P_{\cup}(\X_{\Delta})$, $\tilde{\A}_i(\X_s) \subseteq \A_{i}(\X_s)$ for all $i \in \ZZ$ and $\X_s\supseteq \X_s^\Delta = \tilde{\A}_0(\X_s) \supseteq\tilde{\A}_1(\X_s)\supseteq\tilde{\A}_2(\X_s) \supseteq \cdots$ is a nested sequence of sets in the safe set $\X_s$.
\end{proof}

\begin{theorem}\label{thm:synthesisInvSet}
    Consider the CIS synthesis process. Upon the termination of Algorithm~\ref{al:synthesisInvSet}, indicated by the condition $\tilde{A}_{i+1}(\X_s) = \tilde{A}_{i}(\X_s)$, the output is either the empty set, $\varnothing$, or a CIS, $\C\subseteq\X_s$, expressed as a union of hyperboxes $\C\in\P_{\cup}(\X_\Delta)$.
\end{theorem}

\begin{proof}
    If $\C = \tilde{A}_{i+1}(\X_s) = \tilde{A}_{i}(\X_s) \neq \varnothing$, according to Lemma~\ref{lemma:admissible_nested}, this implies that $\C\subseteq\X_s$, $\C\in\P_{\cup}(\X_\Delta)$ and $\C = \Q(\C, \C)$. Subsequently, following a similar argument as presented in the proof of Theorem~\ref{lemma:admissible}, it can be established that for all $x_k \in \C = \Q(\C, \C)$, there must exist a hyperbox $\I_j^+ = \I (x_k; \tilde{\A}_{i+1}(\X_s))$ and a control $u_k = u (x_k; \tilde{\A}_{i+1}(\X_s) ) \in \mathcal{U}$ defined according to Remark~\ref{rmk:u_cis} ensuring that $x_{k+1} = f(x_k, u_k) \in \F(\I_j^+, \{u_k\}) \subseteq \T = \C$, as established in Theorem~\ref{thm:one_step_return}. Then, $\C$ is a CIS according to Definition~\ref{def:ctrlInv}.
\end{proof}

\begin{theorem}\label{thm:finite_t_termination}
    The CIS synthesis process is finitely-determined. In particular, the set recursion in main Algorithms~\ref{al:synthesisInvSet} will terminate within $N_T = \abs{\Delta(\X_s^{\Delta})}+1$ iterations where $\abs{\Delta(\X_s^{\Delta})}$ is the cardinality of the index set $\Delta(\X_s^{\Delta})\subseteq[n_\Delta]$ defined in Proposition~\ref{prop:ssq:unique}.
\end{theorem}

\begin{proof}
    Firstly, Algorithm~\ref{al:isReturnable} is guaranteed to terminate since every input $\I \in \P_{\cup}(\X_\Delta)$ is a finite collection of hyperboxes $\I_i$.
    Secondly, the \textbf{while}-loop in Algorithm~\ref{al:synthesisUnderQ} will also terminate within a finite number of iterations because the hyperboxes $\I^-_i \subseteq \I^-$ will eventually be reduced to basis hyperboxes. Any hyperboxes that still fail the ${\tt returnable\_verification}$ procedure will remain in $\I^-$ and be removed in the next iteration. This process ensures that $\I^- = \varnothing$, leading to the termination of the algorithm.
    
    Lastly, in Algorithm~\ref{al:synthesisInvSet}, the sequence $\tilde{\A}_0(\X_s) \supseteq \tilde{\A}_1(\X_s) \supseteq \tilde{\A}_2(\X_s) \supseteq \cdots$ forms a nested sequence of sets in $\P_{\cup}(\X_\Delta)$, as established by Lemma~\ref{lemma:admissible_nested}. Equivalently, the index sets defined in Proposition~\ref{prop:ssq:unique}, $\Delta(\tilde{\A}_i(\X_s)) \subseteq [n_\Delta]$ for $i = 0,1,2,\dots$, also form a nested sequence, i.e., $\Delta(\tilde{\A}_0(\X_s)) \supseteq \Delta(\tilde{\A}_1(\X_s)) \supseteq \Delta(\tilde{\A}_2(\X_s)) \supseteq \cdots$.  
    Since Algorithm~\ref{al:synthesisInvSet} terminates when $\tilde{\A}_i(\X_s) = \tilde{\A}_{i+1}(\X_s)$, which is equivalent to $\Delta(\tilde{\A}_i(\X_s)) = \Delta(\tilde{\A}_{i+1}(\X_s))$, we consider the worst-case scenario where the sequence is strictly nested, namely, $\Delta(\tilde{\A}_0(\X_s)) \supsetneq \Delta(\tilde{\A}_1(\X_s)) \supsetneq \Delta(\tilde{\A}_2(\X_s)) \supsetneq \cdots$.
    This implies the following inequality, $\abs{\Delta(\X_s^{\Delta})} = \abs{\Delta(\tilde{\A}_0(\X_s))} > \abs{\Delta(\tilde{\A}_1(\X_s))} > \abs{\Delta(\tilde{\A}_2(\X_s))} > \cdots \geq 0$, where $\tilde{\A}_0(\X_s) := \X_s^{\Delta}$ in Algorithm~\ref{al:synthesisInvSet}.
    This sequence will ultimately reach $|\Delta(\tilde{\A}_i(\X_s))| = 0$ for some $i \in \{1,2,\dots, |\Delta(\X_s^{\Delta})|\}$, implying that $\tilde{\A}_i(\X_s) = \varnothing$. The maximum number of iterations before reaching this condition is equal to $|\Delta(\X_s^{\Delta})|$.  
    Then, at the next iteration $N_T = |\Delta(\X_s^{\Delta})| + 1$, we have $\tilde{\A}_i(\X_s) = \tilde{\A}_{i+1}(\X_s) = \varnothing$, at which point the algorithm terminates.
\end{proof}

\begin{remark}\label{rmk:finer_resol_larger_iter}
    The resolution of the SSQ process can be adjusted through $d_{\min}$; smaller values of $d_{\min}$ result in an SSQ with higher resolution, potentially yielding larger CISs but at the cost of producing a lager $N_T$ and requiring more iterations. 
    We also note that, for a fixed $d_{\min}$, $N_T$ grows exponentially with the state dimensionality. Although the optimization problem \eqref{eq:synthesis} scales linearly with the state space dimension, the overall CIS synthesis process exhibits worst-case exponential time complexity, even though it is conducted offline.
    Nonetheless, we note that the \textbf{for}-loops in the CIS synthesis process can be accelerated through parallel computing, allowing multiple hyperboxes to be processed simultaneously. This approach can significantly reduce computation time, when handling high-dimensional state spaces, by leveraging multi-core processors or distributed computing frameworks.
\end{remark}
\section{Model Predictive Control}\label{sec:methodCtrl}
In this section, we present the results that address Problem~\ref{problem:2}. Specifically, we reformulate Problem~\ref{problem:2} as a mathematical program in Sec.~\ref{method:mpc:raw}, leveraging the results outlined in Sec.~\ref{sec:methodSyn}. We analyse the reformulation's theoretical properties, focusing on aspects such as safety guarantees and recursive feasibility. Additionally, to reduce the computational burden, we develop an alternative MPC formulation with fewer decision variables and constraints that preserves the theoretical properties in Sec.~\ref{method:mpc:misc}. Techniques for warm-starting these mathematical programs with feasible initial guesses are also provided in Sec.~\ref{method:mpc:misc}.

\subsection{Model Predictive Control and Recursive Feasibility}\label{method:mpc:raw}

Given $\X$, $\U$, and $\X_s$ satisfying Assumption~\ref{assume:x_u_box}, the SSQ process, and the CIS synthesis process, we obtain a control invariant set $\C \subseteq \X_s$. Here, $\C \in \P_{\cup}(\X_\Delta)$ is a union of hyperboxes, as established in Theorem~\ref{thm:synthesisInvSet}.  
Hence, with a slight abuse of notation, the CIS $\C$ within the state space $\X$ can be expressed as $\C = \X \backslash \O$, where $\O$ is a finite union of open hyperboxes, $\O = \bigcup_{m=1}^{n_o} ]\underline{o}^{(m)}, \overline{o}^{(m)}[$. Subsequently, Problem~\ref{problem:2} can be formulated as the following Mixed-Integer Linearly Constrained MPC (MILC-MPC) problem.

\begin{subequations}\label{eq:cocp:rec}
\begin{equation}\label{eq:cocp:rec:objective}
    \min\limits_{\substack{
    z^{(i)}_{k+n},~ \hat{z}^{(i)}_{k+n},~\sigma_{k+n}^{(i)},~i=1,\dots,\ell-1,
    \\
    \phi^{(m)}_{k+n},~\psi^{(m)}_{k+n},~m=1,\dots,n_o,
    \\
    x_{k|n+1},~u_{k+n},~n=0,\dots,N-1, 
    }}
    K(x_{k|N}) + \sum_{i=0}^{N-1} l_{k+n}(x_{k|n}, u_{k+n}), 
\end{equation}
\text{subject to: \vspace{-2\baselineskip}}
\begin{equation}\label{eq:cocp:rec:x0}
     x_{k|0} = x_0,
\end{equation}
\begin{equation}\label{eq:cocp:rec:state}
    \underline{x} \leq x_{k|n+1} \leq \overline{x},
\end{equation}
\begin{equation}\label{eq:cocp:rec:ctrl}
    \underline{u} \leq u_{k+n} \leq \overline{u},
\end{equation}
\begin{equation}\label{eq:cocp:rec:nndm:linear}
\left\{
\begin{array}{c}
    x_{k|n+1} = W^{(\ell)}z^{(\ell-1)}_{k+n} +B^{(\ell)}, 
    \\
    \hat{z}^{(i)}_{k+n} = W^{(i)}z^{(i-1)}_{k+n} +B^{(i)},
    ~
    i=1,\dots,\ell-1,   
    \\
    z^{(0)}_{k+n} = [x_{k|n}^T, u_{k+n}^T]^T,
\end{array}
\right.
\end{equation}
\begin{equation}\label{eq:cocp:rec:nndm:relu}
\left\{
\begin{array}{c}
    z_{k+n}^{(i)} \geq 0,
    ~
    z_{k+n}^{(i)} \geq \hat{z}_{k+n}^{(i)},
    \\
    z_{k+n}^{(i)} \leq \hat{z}_{k+n}^{(i)} - \diag(\underline{\hat{z}}^{(i)}_{k+n}) (\one_{n_i\times 1}-\sigma_{k+n}^{(i)}),
    \\
    z_{k+n}^{(i)} \leq \diag(\overline{\hat{z}}^{(i)}_{k+n})\sigma_{k+n}^{(i)},
    \\
    \sigma_{k+n}^{(i)}\in \{0,1\}^{n_i},
    ~
    i=1,\dots,\ell-1,
\end{array}
\right.
\end{equation}
\begin{equation}\label{eq:cocp:rec:safe}
    {\tt MILC\_Inc}
    \brk{
    \hspace{-0.5em}
    \begin{array}{c}
         x_{k|n+1}, x_{k|n+1}, \\
         \{{\phi^{(m)}_{k+n}, \psi^{(m)}_{k+n}}\}_{m=1}^{n_o}
    \end{array}
    \hspace{-0.5em}
    ; 
    \hspace{-0.5em}
    \begin{array}{c}
         \underline{x}, \overline{x}, \X \backslash\C = \X\cap
         \\
         \brk{\cup_{m=1}^{n_o} ]\underline{o}^{(m)}, \overline{o}^{(m)}[}
    \end{array}
    \hspace{-0.5em}
    }
\end{equation}
\begin{equation}\label{eq:cocp:rec:N}
    n=0,\dots,N-1.
\end{equation}
\end{subequations}
In the constraints above, the subscript $*_{k+n}$ refers to decision variables, while $x_{k|n+1}$ denotes the state prediction, corresponding to the $(n+1)${th} prediction horizon, where $n = 0, \dots, N-1$. The vector $\one_{n_i \times 1} \in \RR^{n_i}$ has a dimension of $n_i$ and consists entirely of ones. The decision variable $\sigma_{k+n}^{(i)} \in \{0,1\}^{n_i}$ is a binary vector of dimension $n_i$, where each entry is a binary integer taking values from $\{0,1\}$. Here, $\phi^{(m)}_{k+n}, \psi^{(m)}_{k+n} \in \{0,1\}^{n_x}$ are also binary vectors of dimension $n_x$. 
Similar to \eqref{eq:synthesis:relu:define} and \eqref{eq:rmk_nn_bounds}, vectors $\underline{\hat{z}}^{(i)}_{k+n}$ and $\overline{\hat{z}}^{(i)}_{k+n}$ in \eqref{eq:cocp:rec:nndm:relu} for $n=0,\dots, N-1$ serve as bounds on neuron values $\hat{z}^{(i)}_{k+n} \in [\underline{\hat{z}}^{(i)}_{k+n}, \overline{\hat{z}}^{(i)}_{k+n}]$. These bounds are used to limit the search region and can be recursively computed according to 
\begin{equation}\label{eq:mpc_nn_bounds}
\begin{array}{c}
    \bcur{[\underline{\hat{z}}^{(i)}_{k+n}, \overline{\hat{z}}^{(i)}_{k+n}]}_{i=1}^{\ell}
    \\
    = {\tt reachable\_hyperboxes} ([\underline{\hat{z}}^{(\ell)}_{k+n-1}, \overline{\hat{z}}^{(\ell)}_{k+n-1}], \U),
    \\
    n=0,\dots,N-1,
    ~
    \underline{\hat{z}}^{(\ell)}_{k-1} = \overline{\hat{z}}^{(\ell)}_{k-1}= x_{k|0}.
\end{array}
\end{equation}
The constraints \eqref{eq:cocp:rec:ctrl}, \eqref{eq:cocp:rec:nndm:linear}-\eqref{eq:cocp:rec:nndm:relu}, and \eqref{eq:cocp:rec:safe} enforce the constraints \eqref{eq:cocp:ctrl}, \eqref{eq:cocp:dyna}, and \eqref{eq:cocp:state}, respectively, in the Problem~\ref{problem:2}.

In particular, the constraints \eqref{eq:cocp:rec:state} and \eqref{eq:cocp:rec:ctrl} ensure the feasibility of the control input $u_{k+n}$ and predicted state $x_{k|n+1}$. 
The constraints \eqref{eq:cocp:rec:nndm:linear} represent the linear structures in the NNDS as described in \eqref{eq:nn_structure:out}, \eqref{eq:nn_structure:linear} and \eqref{eq:nn_structure:in}. Meanwhile, the constraints \eqref{eq:cocp:rec:nndm:relu} utilize integer variables to encode the ${\tt ReLU}$ activation function, as originally introduced in \cite{tjeng2017evaluating}. 
With constraints \eqref{eq:cocp:rec:state} implying $x_{k|n+1} \in \X$ and $\X\cap \O=\X\backslash\C$, the remaining safety constraints \eqref{eq:cocp:rec:safe} enforce $x_{k|n+1} \in \C\subseteq \X_s$ by ensuring $x_{k|n+1} \notin \O$, thereby guaranteeing that it remains within the safe set $\X_s$. 
The detailed properties and interpretations of these constraints are provided in the following propositions and their corresponding proofs.

\begin{lemma}\label{lemma:pointwise_relu}
     If a solution to the MILC-MPC \eqref{eq:cocp:rec} exists, then the values of the predicted state trajectory $\{x_{k|n}\}_{n=0}^{N}$ will match the actual state trajectory $\{x_{k+n}\}_{n=0}^{N}$, which evolves according to the control inputs $\{u_{k+n}\}_{n=0}^{N-1}$ and NNDS $f$ in \eqref{eq:nn_structure}, i.e., $x_{k|n+1} = x_{k+n+1} = f(x_{k+n}, u_{k+n})$ for $n = 0, \dots, N-1$, with initial condition $x_{k|0} =x_k = x_0$.
\end{lemma}

\begin{proof}
    With the initial condition $x_{k|0} = x_0$ fixed in \eqref{eq:cocp:rec:x0}, the constraints \eqref{eq:cocp:rec:nndm:linear} precisely represent the linear transformations in the NNDS as defined by \eqref{eq:nn_structure:out}, \eqref{eq:nn_structure:linear} and \eqref{eq:nn_structure:in}. As well-established in \cite{tjeng2017evaluating}, the constraints \eqref{eq:cocp:rec:nndm:relu} is equivalent to ${\tt ReLU}$ activation $z_{k+n}^{(i)} = \max \{0, \hat{z}_{k+n}^{(i)}\}$ defined in \eqref{eq:nn_structure:relu}. Specifically, the $j${th} element of $\sigma^{(i)}_{k+n}$ being zero implies that the $j${th} element of $z_{k+n}^{(i)}$ is also zero, i.e., $\sigma^{(i)}_{k+n,j} = 0\to z_{k+n,j}^{(i)}=0$, meaning that the ${\tt ReLU}$ activation on the $j${th} neuron in the $i${th} layer is inactive. Conversely, if $\sigma^{(i)}_{k+n,j} = 1$, the ${\tt ReLU}$ activation on that neuron is active. 
\end{proof}

\begin{lemma}\label{lemma:cocp_invariance}
    If a solution to the MILC-MPC \eqref{eq:cocp:rec} exists, then the control sequence $\{u_{k+n}\}_{n=0}^{N-1}$, when applied to the dynamics $f$ in \eqref{eq:nn_structure}, ensure that the actual state trajectory satisfies $x_{k+n+1} \in \C\subseteq \X_s$ for all $n = 0, \dots, N-1$.
\end{lemma}

\begin{proof}
Since $x_{k|n} = x_{k+n}$ for all $n = 0, \dots, N$ established in Lemma~\ref{lemma:pointwise_relu}, $[x_{k|n+1},x_{k|n+1}]\subseteq \T$ and $\T=\C$ can be established from Lemma~\ref{lemma:synthesis:subset} with the satisfaction of constraints \eqref{eq:cocp:rec:safe} which implies $x_{k+n+1}=x_{k|n+1}\in\C$. Consequently, the state trajectory remains within the CIS $\C$ and adheres to the safety requirements defined by $\X_s$.
\end{proof}

\begin{theorem}\label{thm:cocp_recursive_feasible}
If the initial state $ x_{k|0} = x_0 \in \C $, where $ \C \subseteq \X_s $, derived from the CIS synthesis process, is a non-empty set, then the optimization problem~\eqref{eq:cocp:rec} is recursively feasible. Namely, there exists an admissible control sequence $ \{u_{k+n}\}_{n=0}^{N-1} \in\U$ that serves as a solution to \eqref{eq:cocp:rec} when $ x_{k|0} = x_0 \in \C $. Furthermore, this control sequence ensures state trajectory $ \{x_{k+n}\}_{n=1}^{N}\in\C $, and \eqref{eq:cocp:rec} remains feasible for all initial states along the trajectory $ \{x_{k+n}\}_{n=1}^{N} $.
\end{theorem}

\begin{proof}
According to Theorem~\ref{thm:synthesisInvSet}, $\C$ is a CIS, which implies that there exists a control input $u_k \in \U$ such that $x_{k|1} = f(x_{k|0}, u_k) \subseteq \C$ when $x_{k|0} \in \C$. This immediately validates the feasibility of constraints \eqref{eq:cocp:rec:state}-\eqref{eq:cocp:rec:nndm:relu}, based on the results of Lemma~\ref{lemma:pointwise_relu}, for the first prediction horizon at $n = 0$. 
Furthermore, $x_{k|1} \in \C$ implies $[x_{k|1},x_{k|1}]\subseteq\T$ with $\T=\C$ leading to the feasibility of constraints \eqref{eq:cocp:rec:safe} for $n = 0$ according to the equivalency established by Lemma~\ref{lemma:synthesis:subset}.
By a similar argument, the feasibility of the constraints at the prediction horizon $n = 1$ can be established, as $x_{k|1} \subseteq \C$ guarantees $x_{k|2} \subseteq \C$. This recursive reasoning not only proves the feasibility of MILC-MPC~\eqref{eq:cocp:rec} for $n=0,\dots,N-1$, but also ensures its feasibility for all $n \in \mathbb{Z}$.
\end{proof}

\subsection{Computational Aspects: Receding-Horizon Scheme and Warm Start}\label{method:mpc:misc}

Finally, we also note that the MILC-MPC~\eqref{eq:cocp:rec} can be further simplified to the following MILC-MPC problem:
\begin{subequations}\label{eq:cocp:simple}
\begin{equation}\label{eq:cocp:simple:objective}
    \min\limits_{\substack{
    z^{(i)}_{k+n},~ \hat{z}^{(i)}_{k+n},~\sigma_{k+n}^{(i)},~i=1,\dots,\ell-1,
    \\
    x_{k|n+1},~u_{k+n},~n=0,\dots,N-1, 
    \\
    \phi^{(m)},~\psi^{(m)},~m=1,\dots,n_o,
    }}
    K(x_{k|N}) + \sum_{i=0}^{N-1} l_{k+n}(x_{k|n}, u_{k+n}), 
\end{equation}
\text{subject to: 
\eqref{eq:cocp:rec:x0},
\eqref{eq:cocp:rec:state},
\eqref{eq:cocp:rec:ctrl},
\eqref{eq:cocp:rec:nndm:linear},
\eqref{eq:cocp:rec:nndm:relu},
\eqref{eq:cocp:rec:N},
}
\begin{equation}\label{eq:cocp:simple:safe}
    {\tt MILC\_Inc}
    \brk{
    \hspace{-0.5em}
    \begin{array}{c}
         x_{k|1}, x_{k|1}, \\
         \{{\phi^{(m)}, \psi^{(m)}}\}_{m=1}^{n_o}
    \end{array}
    \hspace{-0.5em}
    ; 
    \hspace{-0.45em}
    \begin{array}{c}
         \underline{x}, \overline{x}, \X \backslash\C = \X\cap
         \\
         \brk{\cup_{m=1}^{n_o} ]\underline{o}^{(m)}, \overline{o}^{(m)}[}
    \end{array}
    \hspace{-0.5em}
    }.
\end{equation}
\end{subequations}
This formulation inherits all constraints from \eqref{eq:cocp:rec} except for \eqref{eq:cocp:rec:safe}. In particular, instead of enforcing \eqref{eq:cocp:rec:safe} for all $n = 0, \dots, N-1$, the MILC-MPC~\eqref{eq:cocp:simple} enforces it only for the first prediction step $n = 0$ using \eqref{eq:cocp:simple:safe}. This approach reduces the number of decision variables while preserving recursive feasibility and safety guarantee as summarized in Corollary~\ref{coro:mpc}, of which the proof follows from Theorem~\ref{thm:cocp_recursive_feasible}.

\begin{corollary}\label{coro:mpc}
In a receding-horizon optimization scheme, a control sequence $\{u_{k+n}\}_{n=0}^{N-1}$ is optimized at the current time step $k$, the first control action $u_k$ is applied, and a new control sequence is computed at the next step. 
If the initial state $ x_{k|0} = x_0 \in \C $, where $ \C \subseteq \X_s $, derived from the CIS synthesis process, is a non-empty set, then the optimization problem~\eqref{eq:cocp:simple} is recursively feasible. Namely, there exists an admissible control sequence $ \{u_{k+n}\}_{n=0}^{N-1} \in\U$ that serves as a solution to \eqref{eq:cocp:simple} when $ x_{k|0} = x_0 \in \C $. Furthermore, this control sequence ensures the subsequent state $ x_{k+1}\in\C $, and \eqref{eq:cocp:simple} remains feasible for initial state $x_{k+1}$.
\end{corollary}

In addition, given an initial state $x_0$ and leveraging the byproduct control law $u_k = \pi(x_k ; \C)$ according to Remark~\ref{rmk:u_cis}, we can warm-start the MILC-MPC \eqref{eq:cocp:simple} with the following feasible initial guess for the decision variables.
\begin{multline}\label{eq:warmStart}
\left\{
\begin{array}{c}
    z^{(i)}_{k+n},~ \hat{z}^{(i)}_{k+n},~\sigma_{k+n}^{(i)},~i=1,\dots,\ell-1,
    \\
    x_{k|n+1},~u_{k+n},~n=0,\dots,N-1, 
    \\
    \phi^{(m)},~\psi^{(m)},~m=1,\dots,n_o
\end{array}
\right\}
=
\\
{\tt warmStart}(x_0)
:=
\\
\left\{
\begin{array}{c}
    x_{k|n+1} = W^{(\ell)}z^{(\ell-1)}_{k+n} +B^{(\ell)}, 
    \\
    z^{(i)}_{k+n} = \max\{0, \hat{z}^{(i)}_{k+n}\},
    ~
    \sigma_{k+n}^{(i)} = \IversonBracket{\hat{z}^{(i)}_{k+n} \geq 0 }
    \\
    \hat{z}^{(i)}_{k+n} = W^{(i)}z^{(i-1)}_{k+n} +B^{(i)},
    ~
    i=1,\dots,\ell-1,   
    \\
    z^{(0)}_{k+n} = [x_{k|n}^T, u_{k+n}^T]^T,
    ~
    u_{k+n} = \pi(x_{k|n};\C), 
    \\
    n=0,\dots,N-1, 
    ~
    x_{k|0} = x_0,
    \\
    \phi^{(m)} = \IversonBracket{x_{k|1} \leq \underline{o}^{(m)}},
    ~
    \psi^{(m)} = \IversonBracket{x_{k|1} \geq \overline{o}^{(m)}},
    \\
    m=1,\dots,n_o,
\end{array}
\right.
\end{multline}
where the Iverson bracket $\IversonBracket{a \leq b} \in \mathbb{R}^p$ returns a binary vector for $a, b \in \mathbb{R}^p$, and the $j$th dimension is 1 if the condition $a_j \leq b_j$ is true and 0 otherwise. 
This initial guess is feasible because $u_k = \pi(x_k ; \C)$ renders the set $\C$ a CIS, as discussed in Remark~\ref{rmk:u_cis}, and the integer variables $\sigma_{k+n}^{(i)}$ encode the neuron activation status discussed in the proof of Lemma~\ref{lemma:pointwise_relu}, while $\phi^{(m)}$ and $\psi^{(m)}$ have implications established in \eqref{eq:phi_psi_implication}.  
Using branch-and-bound optimization with this initial guess, we can establish an upper bound on the optimal cost \eqref{eq:cocp:simple:objective} and potentially limit the number of branch-and-bound nodes, thereby reducing the computational burden. A similar procedure can be adopted for warm starting \eqref{eq:cocp:rec}.

\begin{remark}    
The MILC-MPCs \eqref{eq:cocp:rec} and \eqref{eq:cocp:simple} involves linear constraints with integer variables. Depending on the choice of cost functions $K$ and $l_{k+n}$, the MILC-MPCs \eqref{eq:cocp:rec} and \eqref{eq:cocp:simple} can be formulated as either a MILP if the cost functions are defined using $L_{1}$ or $L_{\infty}$ norms, or as a MLQP if the cost functions are defined using the $L_{2}$ norm.
\end{remark}

\begin{remark}\label{rmk:cocprec:num_decision_var}
In MILC-MPCs~\eqref{eq:cocp:rec} and \eqref{eq:cocp:simple}, the total number of continuous decision variables is $N(n_x + n_u + 2\sum_{i=1}^{\ell-1} n_i)$, while the total number of integer or binary decision variables is $N(2n_o + \sum_{i=1}^{\ell-1} n_i)$ and $2n_o + N\sum_{i=1}^{\ell-1} n_i$, respectively. Importantly, both quantities scale linearly with the number of neurons in NNDS and the dimensionality of state and control variables.
\end{remark}

\begin{remark}\label{rmk:cocprec:delete_constraints}
Similar to Remark~\ref{rmk:inv:delete_constraints}, to further reduce the number of decision variables, majorities of the open hyperboxes $\obsq{\underline{o}^{(m)}, \overline{o}^{(m)}}$ can be eliminated from the union $\O$ when formulating constraints \eqref{eq:cocp:simple:safe} and  \eqref{eq:cocp:rec:safe} for a given $n${th} prediction horizon. 
This simplification is valid if it is determined that $\bar{\F}^{(n)}(\{x_{k|0}\}, \U) \cap \bsq{\underline{o}^{(m)}, \overline{o}^{(m)}} = \varnothing$, with
\begin{equation*}
    \bar{\F}^{(n)}(\X_k, \U_k) := 
    \underbrace{
    \bar{\F}\Big( \bar{\F}\big(\cdots\bar{\F}
    }_\text{apply $\bar{\F}$ defined in \eqref{eq:reachability} $n$ times}
    (\X_k,\U_k),\U_k\big),\U_k\Big).
\end{equation*}
\end{remark}

\section{Numerical Results}\label{sec:results}
We demonstrate the proposed methods for a lane-keeping case study. The configuration of the lane-keeping problem is presented in Sec.~\ref{result:problem}. We synthesize CISs offline to ensure the safety of vehicle motion subject to lane-keeping constraints in Sec.~\ref{result:synthesis}, while an MPC is formulated to stabilize the vehicle online to user-defined set points in Sec.~\ref{result:control}.

\subsection{Lane Keeping Scenario}\label{result:problem}
\begin{figure}[ht!]
\begin{center}
\includegraphics[width=0.99\linewidth]{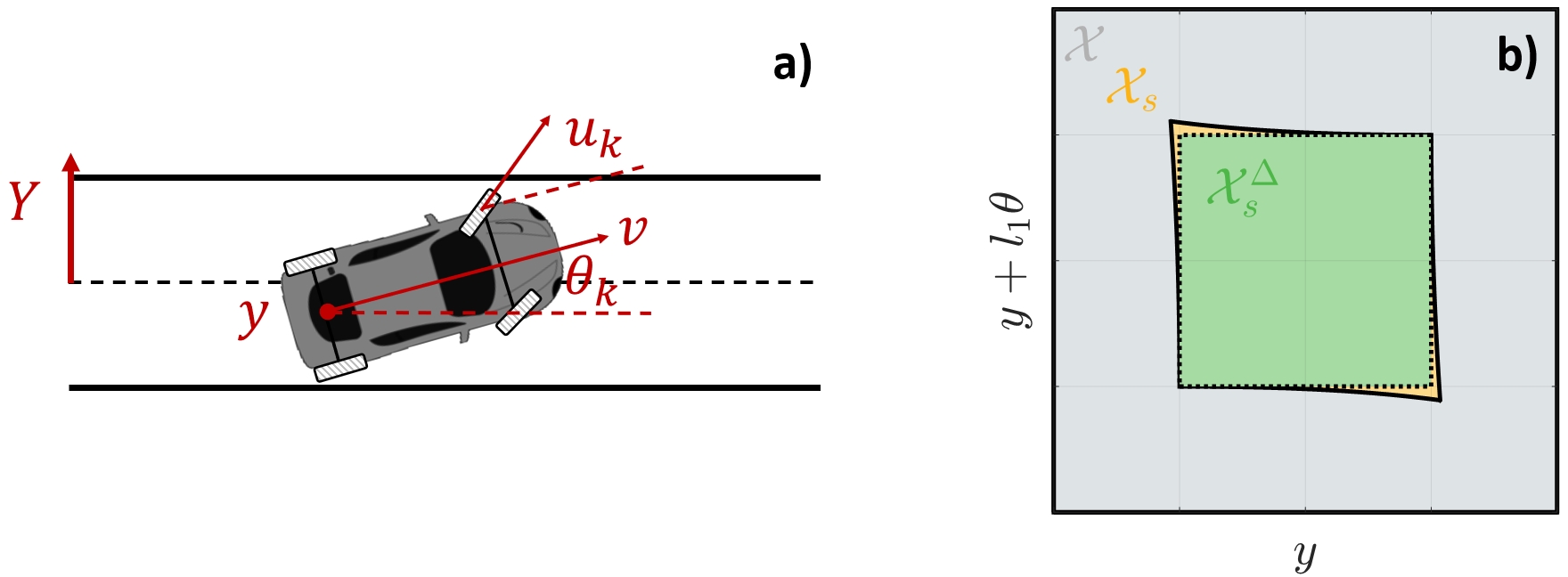}
\end{center}\vspace{-1.8em}
\caption{A lane-keeping problem: a) A vehicle traveling at a constant speed is controlled through steering to ensure it remains within the lane boundaries.  
b) The safe set $\X_s$ is approximated from within using a hyperbox, denoted as $\X_s^\Delta$, in the $y-(y + l_1\theta)$ plane.}
\label{fig:laneKeepingProblem}
\end{figure}

We consider an NNDS with two hidden layers of 8 and 4 neurons, i.e., $\ell = 3$, $n_1 = 8$ and $n_2 = 4$. For a vehicle with length $l_1$ and width $l_2$ (see Fig.~\ref{fig:laneKeepingProblem}a), the NNDS learns the following vehicle kinematics model,
\begin{equation}
    x_{k+1} 
    = f(x_k, u_k),
    ~
    x_k = 
    \begin{bmatrix}
        y_{k}\\
        y_{k} + l_1\theta_{k}
    \end{bmatrix},
\end{equation}
which imitates the following discrete-time bicycle model,
\begin{multline}\label{eq:bicycle}
    x_{k+1} 
    = f_\text{bicycle}(x_k, u_k),
    \\
    y_{k+1} = y_{k} + vdt\sin\theta_k ,
    ~
    \theta_{k+1} = \theta_{k} + \frac{vdt}{l_1} \tan u_k,
\end{multline}
where $y_k$ is the lateral position of the center of the vehicle rear wheel axle; $\theta_k$ is the vehicle yaw angle; $u_k$ is the vehicle steering angle; $v>0$ is the vehicle speed and is assumed to be a constant; $dt$ is the time elapsed between two subsequent discrete time steps $k$ and $k+1$. In the lane-keeping problem, we focus on the vehicles' lateral behavior, excluding the longitudinal position from the states. Meanwhile, the state space and control admissible set are defined according to
\begin{equation}
\begin{array}{c}
\X = [-(w-l_2), (w-l_2)]\times[-\frac{\pi}{2},\frac{\pi}{2}],
~
\U = [-u_{\max}, u_{\max} ],
\end{array}
\end{equation}
where $u_{\max}$ is the maximum allowed steering angle in radians and $w$ is the lane width in meters.

We adopt a lane-keeping configuration similar to that in \cite{laneKeeping_LX, laneKeeping_GB}, where the vehicle is modeled as a rectangular shape that must remain within the lower and upper lane boundaries. These boundaries are visualized as two parallel black lines in Fig.~\ref{fig:laneKeepingProblem}a. As shown in Fig.~\ref{fig:laneKeepingProblem}b, this geometric consideration defines the following safe set.
\begin{equation}
\X_s =
\left\{
\begin{bmatrix}
    y
    \\
    y+l_1\theta
\end{bmatrix}
    :
    \hspace{-0.5em}
\begin{array}{c}
    y+\frac{1}{2}l_2\cos\theta\leq \frac{1}{2}w,
    \\ 
    y+l_1\sin\theta+\frac{1}{2}l_2\cos\theta \leq \frac{1}{2}w,
    \\
    y-\frac{1}{2}l_2\cos\theta\geq -\frac{1}{2}w,
    \\ 
    y+l_1\sin\theta-\frac{1}{2}l_2\cos\theta \geq -\frac{1}{2}w
\end{array}
    \hspace{-0.5em}
\right\},
\end{equation}
which can be approximated from within using the following safe hyperbox,
\begin{equation}
    \X_s^\Delta = [-\frac{1}{2}(w-l_2), \frac{1}{2}(w-l_2)] \times [-\frac{1}{2}(w-l_2), \frac{1}{2}(w-l_2)].
\end{equation}
In the sequel, we use the following parameters for experiments: $l_1=5~\rm{m}$, $l_2=2~\rm{m}$, $w=3.5~\rm{m}$, $v=6~\rm{m/s}$, and $dt=0.1~\rm{sec}$.

\begin{remark}
We densely sample a training dataset $\{(x_k, u_k, x_{k+1})_i\}_{i=1}^{1\times 10^6}$ from $\X \times \U$ using the vehicle bicycle model \eqref{eq:bicycle}. Using the PyTorch library \cite{pytorch}, we train the NNDS with the Stochastic Gradient Descent algorithm on this dataset to minimize the squared error, defined as $\norm{f_\text{bicycle}(x_k, u_k) - f(x_k, u_k)}^2_2$. The errors are empirically evaluated across $\X \times \U$, yielding a value smaller than $ 1.2\times 10^{-4}$.  
Our focus is on the empirical demonstration of the algorithmic developments on the NNDS $f$, rather than on $f_\text{bicycle}$. The introduction of the bicycle model serves only to provide the NNDS with physical interpretability. Such model mismatches will be considered in future investigations.
\end{remark}

\subsection{CIS Synthesis}\label{result:synthesis}

\begin{figure*}[ht!]
    \centering
    \includegraphics[width=0.92\textwidth]{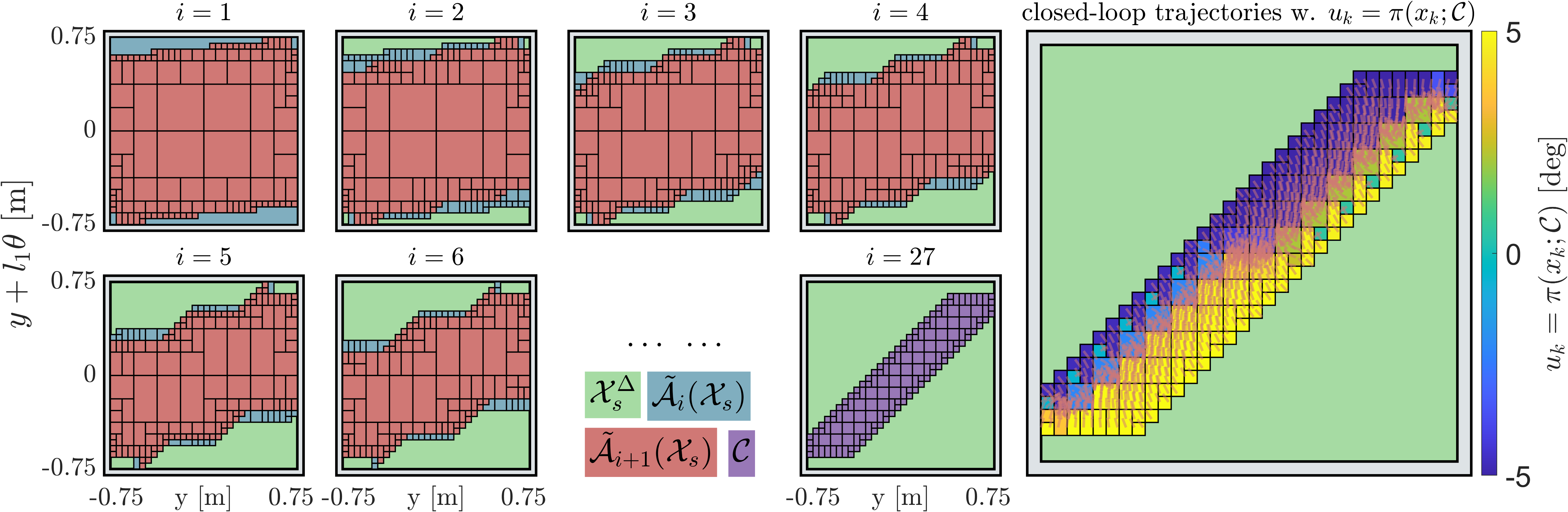}\vspace{-0.5em}
    \caption{
    The CIS synthesis process is conducted with $ d_{\min} = (w - l_2)/32$ m and $ u_{\max} = 5^\circ $. 
    On the left, the proposed method computes $ (i+1) $-Step Admissible sets $ \tilde{\A}_{i+1}(\X_s) $ (represented by unions of red cubes) using $ i $-Step Admissible sets $ \tilde{\A}_{i}(\X_s) $ (represented by unions of blue cubes), ultimately producing a nested sequence of $ i $-Step Admissible sets. The iterative process from iteration 1 to 6 is visualized, and the process terminates at iteration 27, generating a CIS represented by unions of purple cubes.  
    On the right, as a byproduct of this process, we obtain a control law $ u_k = \pi(x_k ; \C) $ according to Remark~\ref{rmk:u_cis}. This control law is visualized in a heatmap, empirically demonstrating how it renders the trajectories (red dashed lines) of the closed-loop NNDS $ f $ forward invariant with respect to $ \C $.
    }
    \label{fig:synthesisCIS}
\end{figure*}

\begin{figure}[ht!]
\begin{center}
\includegraphics[width=0.99\linewidth]{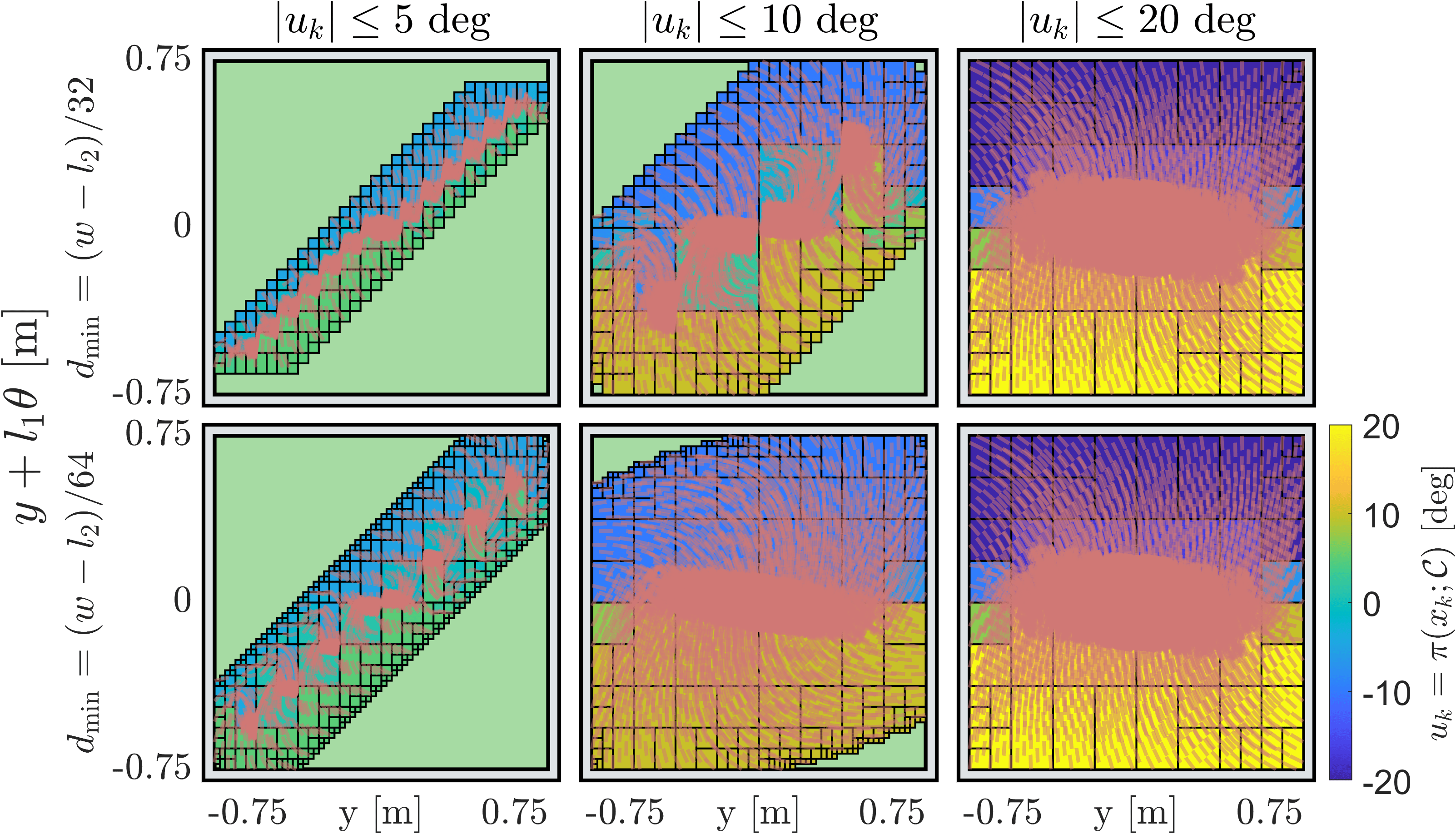}
\end{center}\vspace{-1.5em}
\caption{Synthesizing CISs with different configurations of the SSQ resolution $ d_{\min} $ in the SSQ process and steering limits $ u_{\max} $. Similarly, control laws $ u_k = \pi(x_k ; \C) $ are visualized in heatmaps, with trajectories of the closed-loop NNDS $f$ represented by red dashed lines.}
\label{fig:controlMaps}
\end{figure}

Fig.~\ref{fig:synthesisCIS} illustrates the CIS synthesis process.  
We first apply the SSQ process with $ d_{\min} = (w - l_2)/32 $ m, yielding $ \abs{\Delta(\X_s^{\Delta})} = \ceil{(w - l_2) / d_{\min}}^2 = 1024 $, and the upper bound on the number of iterations is $ N_T = \abs{\Delta(\X_s^{\Delta})} + 1 = 1025 $ according to Theorem~\ref{thm:finite_t_termination}. The control constraint is $ u_{\max} = 5^\circ $.  
The CIS synthesis process generates a nested sequence of $ i $-Step Admissible Sets $ \tilde{\A}_i(\X_s) $, empirically validating Lemma~\ref{lemma:admissible_nested}. Notably, the process terminates at iteration number $ 27 \ll N_T $, which is significantly smaller than the established theoretical upper bound $ N_T = 1025 $.  
Additionally, the CIS synthesis process derives a control law $ u_k = \pi(x_k; \C) $ as a byproduct, according to Remark~\ref{rmk:u_cis}. As shown in the heatmap in Fig.~\ref{fig:synthesisCIS}, this control law keeps the closed-loop NNDS trajectories within the CIS $ \C $, empirically validating Theorem~\ref{thm:synthesisInvSet}.

\begin{figure}[ht!]
    \centering
    \includegraphics[width=0.99\linewidth]{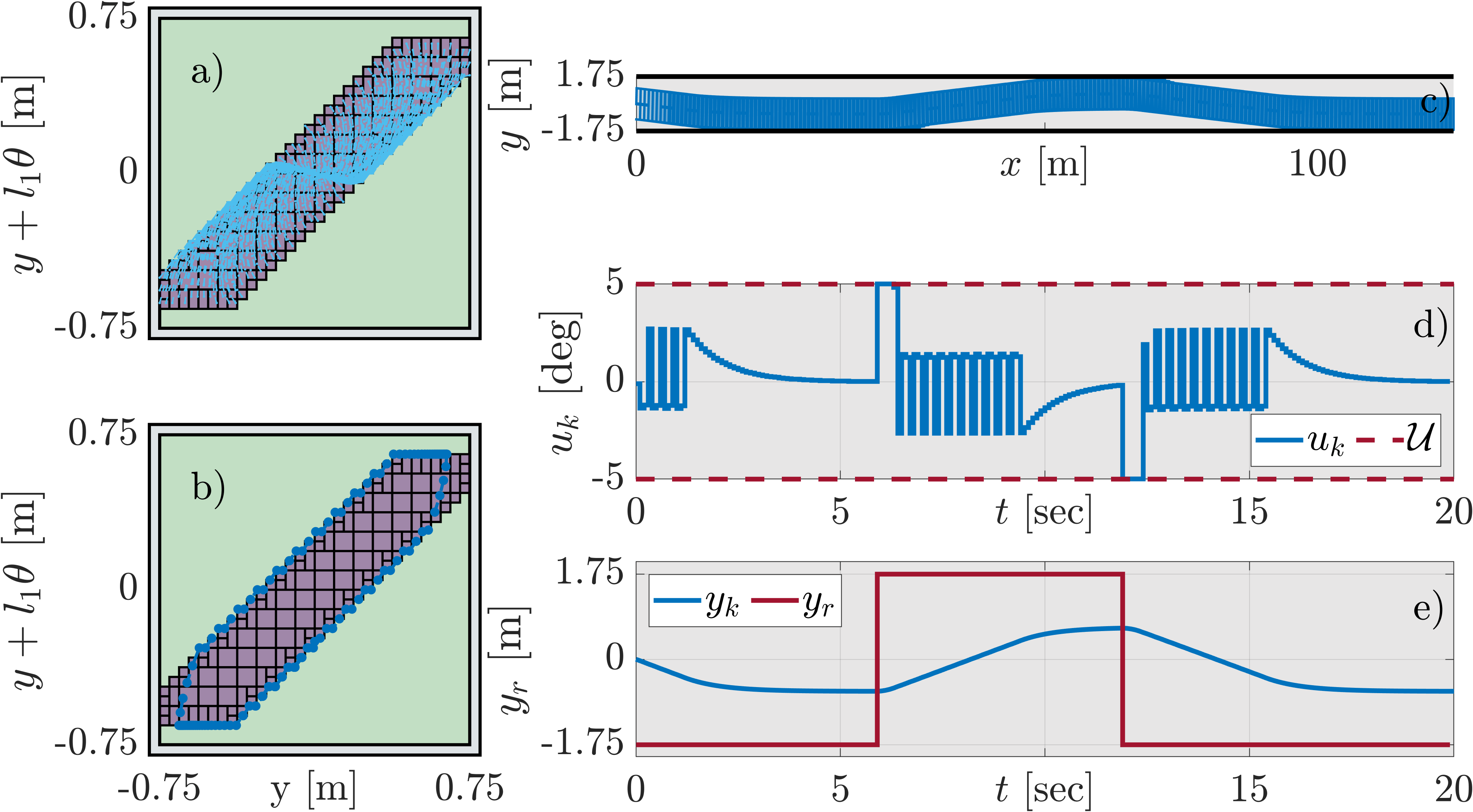}\vspace{-1.5em}
    \caption{Closed-loop simulations with the MPC controller: a) Closed-loop trajectories of the MPC regulating the system to the origin from different initial states inside $\C$. b)-e) The MPC tracks an extreme reference command, alternating between the right lane boundary and the left lane boundary while keeping the system inside $\C$. b) State-space plot showing system evolution. c) Vehicle trajectories in the straight lane. d) Visualization of the control signal and constraints. e) Comparison of the commanded lateral reference with the actual vehicle lateral coordinates.}
    \label{fig:mpcSim}
\end{figure}

In addition, the objective function \eqref{eq:synthesis:obj} of MILP-OP~\eqref{eq:synthesis} is defined as $ \norm{\underline{x}_{k+1} + \overline{x}_{k+1}}_2^2 $, which regulates the next state toward the origin, as also illustrated by the trajectories in Fig.~\ref{fig:synthesisCIS}. The MILP-OP~\eqref{eq:synthesis} is implemented in MATLAB using the YALMIP interface \cite{yalmip} with the Gurobi MIQP solver \cite{gurobi}.  
The computations are performed on a 13th Gen Intel i9-13900F CPU with 32 GB RAM. During the entire CIS synthesis process shown in Fig.~\ref{fig:synthesisCIS}, the optimization solver is called 5478 times to solve MILC-OP~\eqref{eq:synthesis}, with each solution taking an average of $ 0.0765 \pm 0.0011 $ seconds.

Meanwhile, we also repeat the above procedure with different values of $ d_{\min} $ and $ u_{\max} $. The results are summarized in Fig.~\ref{fig:controlMaps}. With a smaller $ d_{\min} = (w - l_2)/64 $, the SSQ process generates smaller basis hyperboxes, and results in a larger CIS $ \C $, as previously noted in Remark~\ref{rmk:finer_resol_larger_iter}.  
Similar observations hold when relaxing the control constraints by increasing $ u_{\max} $. Intuitively, increasing $ u_{\max} $ provides the closed-loop system with greater control authority, allowing for a larger CIS $ \C $. Particularly, with $ u_{\max} = 20^\circ $, the process yields a CIS $ \C $ that is equal to the safe hyperbox $ \X_s^{\Delta} $.  
Additionally, for the results in the first row of Fig.~\ref{fig:controlMaps}, the CIS synthesis processes using $ d_{\min} = (w - l_2)/32 $ m terminate at iteration numbers 27, 16, and 1, respectively, all of which are smaller than the upper bound $ N_T = 1025 $. A similar trend is observed for the results in the second row of Fig.~\ref{fig:controlMaps}, where the CIS synthesis processes terminate at iteration numbers 48, 7, and 1, respectively, with a smaller $ d_{\min} = (w - l_2)/64 $ m. The number of iterations is also significantly smaller than the upper bound $ N_T = \ceil{(w - l_2) / d_{\min}}^2  +1 = 4097 $.  
However, with a smaller $ d_{\min} $, the CIS synthesis process yields larger CISs, which may require more iterations. For example, when $ u_{\max} = 5^\circ $, the iteration count increases from 27 to 48 as $ d_{\min} $ decreases. Demonstration videos are available at \url{https://bit.ly/43huSTr}.


\subsection{Model Predictive Control}\label{result:control}

In this section, we use the CIS synthesized with $ d_{\min} = (w - l_2)/32 $ and $ u_{\max} = 5^\circ $, while keeping other experimental settings unchanged. Additionally, we solve the MILC-MPC \eqref{eq:cocp:simple} with the objective function  
$\frac{1}{2} (x_{k|N} - x_r)^TQ(x_{k|N} - x_r) + \frac{1}{2}\sum_{n=0}^{N-1} \brk{ x_{k|n}^TQx_{k|n} + u_{k+n}^TRu_{k+n}}$,
where the prediction horizon is set to $ N=5 $, the weighting matrices are defined as $ Q=\text{diag}([2~2]^T) $ and $ R = 1 $, and $ x_r = [y_r~ y_r + l_1\theta_r]^T \in \mathbb{R}^2 $ represents a user-defined tracking reference. Under this formulation, the MILC-MPC \eqref{eq:cocp:simple} is a MIQP.  

As shown in Fig.~\ref{fig:mpcSim}a, we first command the vehicle to track the origin $ y_r=\theta_r=0 $ from different initial states inside the CIS $ \C $. The closed-loop trajectories, each with a simulation time of 20 seconds, converge to the origin while staying in the CIS $ \C \subseteq \X_s $, as shown by cyan trajectories. The solution is recursively feasible and provides improved convergence performance to the origin compared to the control laws $u_k = \pi(x_k ; \C)$ in Fig.~\ref{fig:synthesisCIS} and Fig.~\ref{fig:controlMaps}.  
Additionally, we note that the converted MIQP can be solved in an average computation time of $0.0074 \pm 0.0010$ seconds without using ${\tt warmStart}$ in \eqref{eq:warmStart}, which is attributed to the reduction of decision variables described in Remarks~\ref{rmk:cocprec:num_decision_var} and~\ref{rmk:cocprec:delete_constraints}. The ${\tt warmStart}$ feature may become useful for larger-scale problems, potentially further reducing computation time.


Meanwhile, in Fig.~\ref{fig:mpcSim}b-e, we command the vehicle to track a step lateral reference $ y_r = y_r(t) $, where the vehicle alternates between tracking the right and left lane boundaries. The proposed method successfully tracks the lane boundaries with high accuracy, as shown in Fig.~\ref{fig:mpcSim}c, while ensuring that control constraints are satisfied, as illustrated in Fig.~\ref{fig:mpcSim}d. Furthermore, the vehicle remains within the CIS $ \C $, as depicted in Fig.~\ref{fig:mpcSim}b, thereby ensuring that it stays within the lane boundaries throughout the simulation.

\section{Conclusion and Future Work}\label{sec:conclusion}
In this paper, we have proposed an algorithmic framework that addressed critical challenges in ensuring safety and feasibility in control dynamical systems modeled by neural networks. By developing algorithms for synthesizing a CIS via finite-step set recursion, we established theoretical guarantees for closed-loop forward invariance and safety. 
In addition, we proposed an MPC framework that integrates these invariant sets using MILCs, enabling robust and recursively feasible online control. 
Simulation results in an autonomous driving application demonstrate that the proposed methods effectively ensure closed-loop system safety, maintain recursive feasibility of the MPC, and provide robust performance suitable for safety-critical control applications. 
Future research will investigate improvements in the time complexity of the offline CIS synthesis process, particularly for high-dimensional systems, and extend these methods to handle uncertainties.
\section*{References}\vspace{-1.5em}
\bibliographystyle{IEEEtran}
\bibliography{ref}
\vspace{-4em}
\begin{IEEEbiography}[{\includegraphics[width=1in,height=1.25in,clip,keepaspectratio]{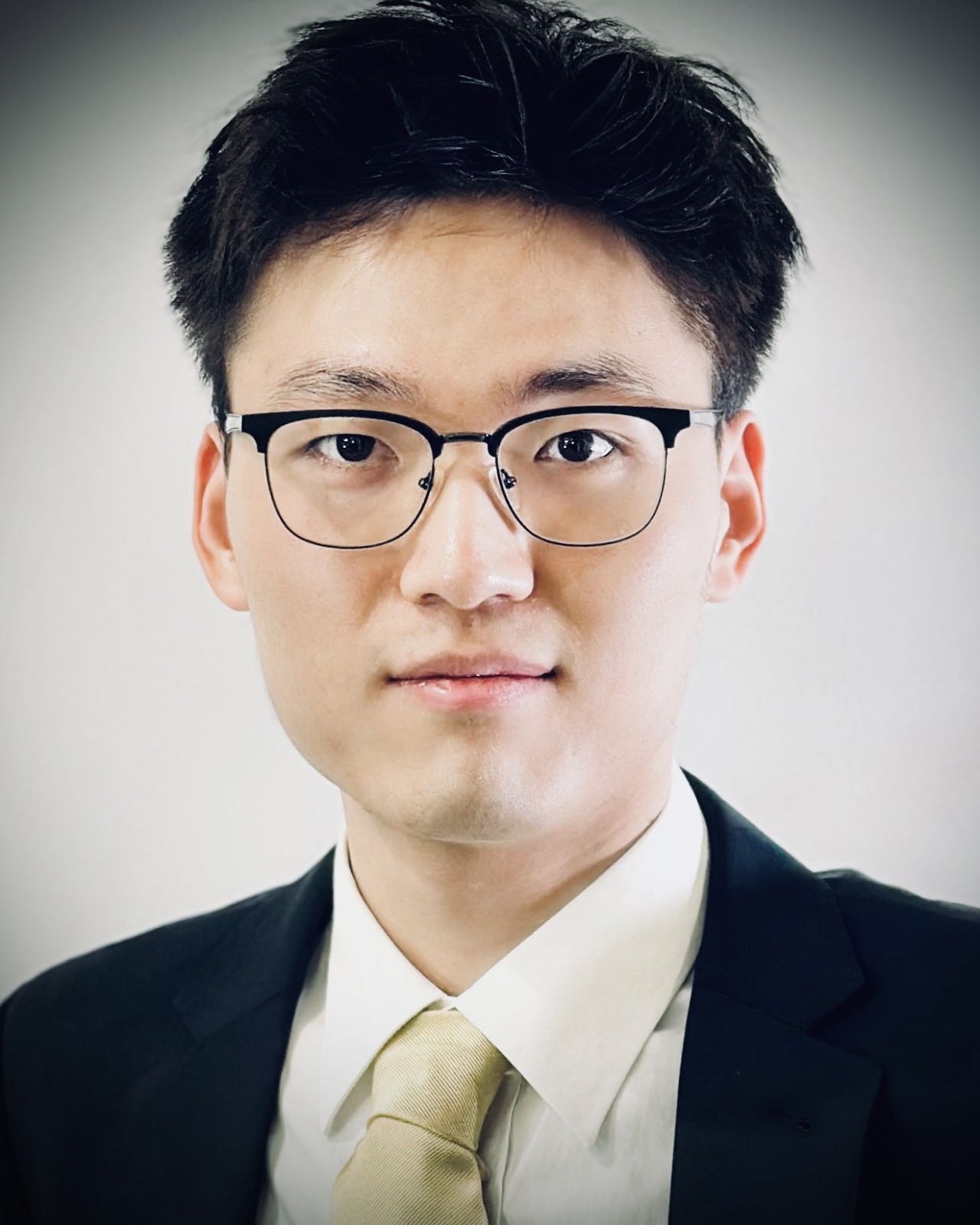}}]{Xiao Li}
received the B.S. degree in mechanical engineering from Shanghai Jiao Tong University, Shanghai, China, in 2019, and the M.S. degree in mechanical engineering from the University of Michigan, Ann Arbor, MI, USA, in 2021, where he is currently pursuing Ph.D. degree in aerospace engineering. His research interests include learning-based methods in constrained optimization and in human-in-the-loop decision-making for autonomous agents.
\end{IEEEbiography}

\vspace{-4em}
\begin{IEEEbiography}[{\includegraphics[width=1in,height=1.25in,clip,keepaspectratio]{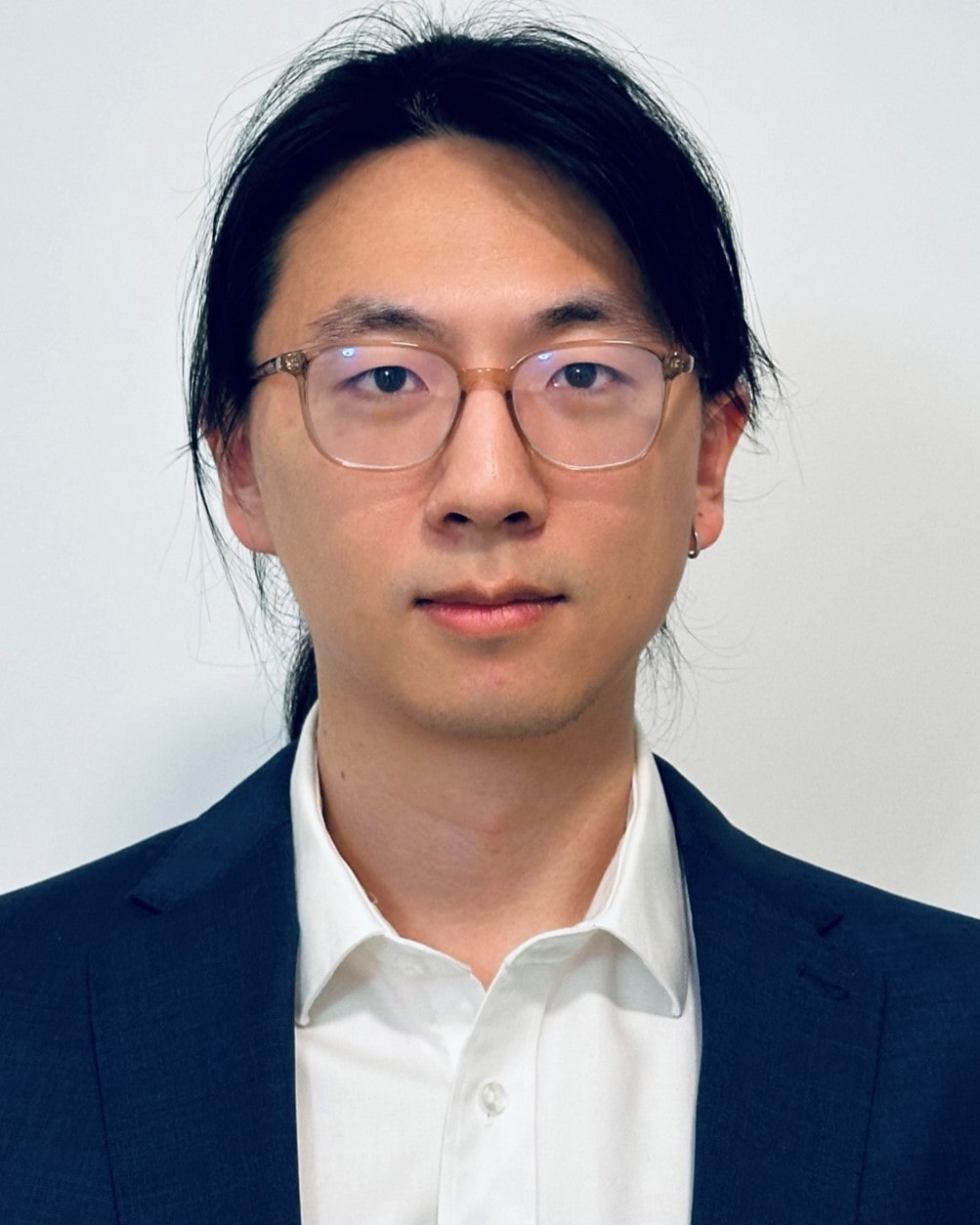}}]{Tianhao Wei}
Tianhao Wei received the B.S. degree in Computer Science from Zhejiang University in 2019, and the Ph.D. degree in Electrical and Computer Engineering (affiliated with the Robotics Institute) from Carnegie Mellon University in 2024. His research interests include safe and intelligent robotic systems and formal safety guarantees for control and neural networks.
\end{IEEEbiography}

\vspace{-4em}
\begin{IEEEbiography}[{\includegraphics[width=1in,height=1.25in,clip,keepaspectratio]{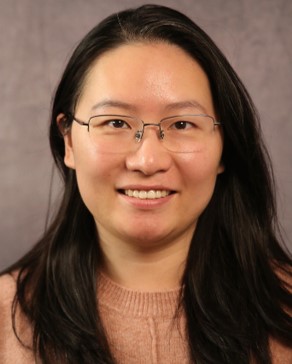}}]{Changliu Liu}
Changliu Liu is an assistant professor in the Robotics Institute, School of Computer Science, Carnegie Mellon University (CMU), where she leads the Intelligent Control Lab. She received her Ph.D. degree from University of California at Berkeley in 2017. Her research interests lie in the design and verification of human-centered intelligent systems with applications to manufacturing and transportation and on various robot embodiments, including robot arms, mobile robots, legged robots, and humanoid robots. Her work has been recognized by NSF Career Award, Amazon Research Award, Ford URP Award, Advanced Robotics for Manufacturing Champion Award, Young Investigator Award at International Symposium of Flexible Automation, and many best/outstanding paper awards.
\end{IEEEbiography}

\vspace{-4em}
\begin{IEEEbiography}[{\includegraphics[width=1in,height=1.25in,clip,keepaspectratio]{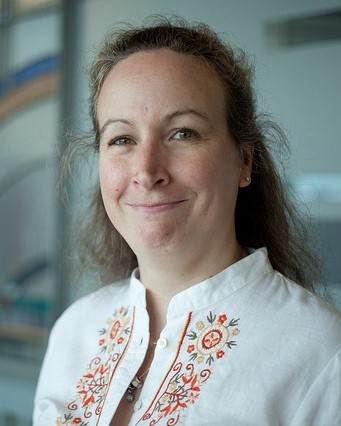}}]{Anouck Girard}
received the Ph.D. degree in ocean engineering from the University of California, Berkeley, CA, USA, in 2002. She has been with the University of Michigan, Ann Arbor, MI, USA, since 2006, where she is currently a Professor of Robotics and Aerospace Engineering, and the Director of the Robotics Institute. Her current research interests include vehicle dynamics and control, as well as decision systems. Dr. Girard was a Fulbright Scholar in the Dynamic Systems and Simulation Laboratory at the Technical University of Crete in 2022.
\end{IEEEbiography}

\vspace{-4em}
\begin{IEEEbiography}[{\includegraphics[width=1in,height=1.25in,clip,keepaspectratio]{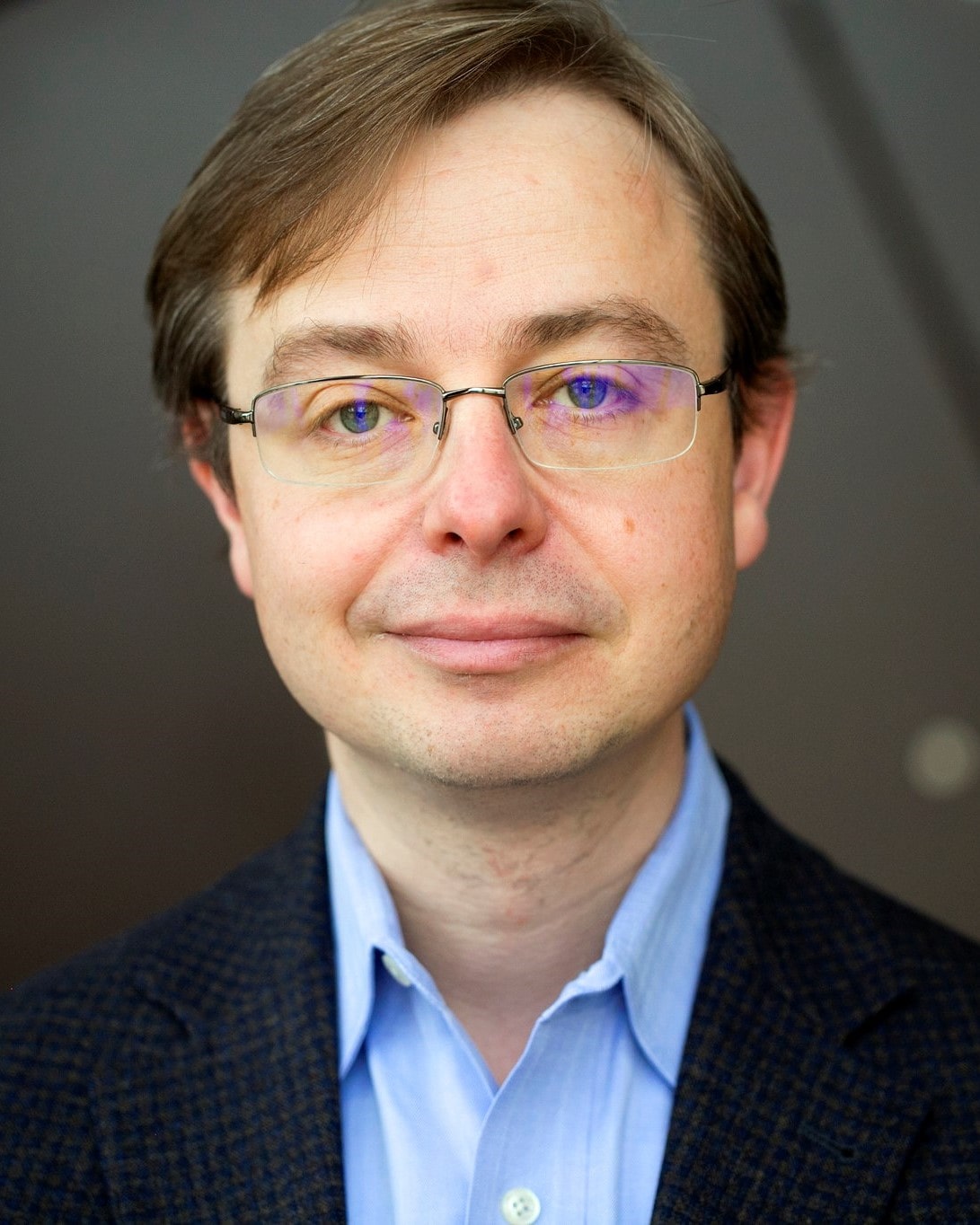}}]{Ilya Kolmanovsky}
received his Ph.D. degree in Aerospace Engineering in 1995 from the University of Michigan, Ann Arbor. He is presently a Pierre T. Kabamba Collegiate Professor in the Department of Aerospace Engineering at the University of Michigan.  His research interests are in control theory for systems with state and control constraints and in aerospace and automotive control applications.
\end{IEEEbiography}

\end{document}